\documentclass{article}
\title{Generalized Alpha Investing: Definitions, Optimality Results, and
Application to Public Databases}

\author{Ehud Aharoni and Saharon Rosset\\IBM Research Laboratory in Haifa,
         Israel and Tel Aviv University, Israel\\
Email: aehud@il.ibm.com, saharon@post.tau.ac.il}

\usepackage{amsfonts}
\usepackage[dvips]{graphicx}
\usepackage{subfigure}
\usepackage{fullpage}
\newtheorem{theorem}{Theorem}[section]
\newtheorem{lemma}[theorem]{Lemma}
\newtheorem{proposition}[theorem]{Proposition}
\newtheorem{corollary}[theorem]{Corollary}
\newtheorem{assumption}[theorem]{Assumption}
\newtheorem{definition}[theorem]{Definition}

\newenvironment{proof}[1][Proof]{\begin{trivlist}
\item[\hskip \labelsep {\bfseries #1}]}{\end{trivlist}}
\newcommand{\qed}{\nobreak \ifvmode \relax \else
      \ifdim\lastskip<1.5em \hskip-\lastskip
      \hskip1.5em plus0em minus0.5em \fi \nobreak
      \vrule height0.75em width0.5em depth0.25em\fi}

\newcommand{\ccj}{\varphi_j}

\newcommand{\ddj}{\psi_j}

\newcommand{\ppj}{\rho_j}
\newcommand{\II}{{\cal I}}
\newcommand{\JJ}{{\cal J}}
\newcommand{\QPDAS}{QPD\mbox{-}AS}
\newcommand{\QPDASR}{QPD\mbox{-}ASR}
\newcommand{\QPDASROPT}{QPD\mbox{-}ASR\mbox{-}OPT}
\newcommand{\LN}{L:\mathbb{N}\to\mathbb{R}}
\newcommand{\ero}{ERO}

\newcommand{\citep}{\cite}

\begin{document}
\maketitle

\begin{abstract}
The increasing prevalence and utility of large, public databases
necessitates the development of appropriate methods for controlling
false discovery. Motivated by this challenge, we discuss the generic
problem of testing a possibly infinite stream of null hypotheses. In
this context, Foster and Stine (2008) suggested a novel method named Alpha Investing
for controlling a false discovery measure known as mFDR.
We develop a more general procedure for controlling mFDR, of which Alpha
Investing is a special case. We show that in common, practical situations,
the general procedure can be optimized to produce an expected reward optimal (ERO) version,
which is more powerful than Alpha Investing.

We then present the concept of quality preserving databases (QPD),
originally introduced in Aharoni et al. (2011), which formalizes
efficient public database management to simultaneously save costs and
control false discovery. We show how one variant of generalized alpha
investing can be used to control mFDR in a QPD and lead to significant
reduction in costs compared to na\"{\i}ve approaches for controlling the
Family-Wise Error Rate implemented in Aharoni et al. (2011).
\end{abstract}

\section{Introduction}

As the extent of available data and scientific questions being
addressed through hypothesis testing continues to increase,
statisticians are required to develop appropriate approaches that
control false discoveries in large scale testing  to assure statistical
validity of discoveries and publications, while facilitating efficient
use of the data and assuring maximal power to the studies.
For example, modern genome-wide association studies already perform on
the order of one million tests per study. Issues of power, sample size
and false discoveries in such
studies are recognized as a major concern~\citep{gwas_review}.

This situation has given rise to a vast and growing literature on
multiple comparison control and corrections, encompassing a variety
of methods for controlling a range of measures of false discovery.
Let $R$ denote the number of rejected null hypotheses (discoveries),
and $V$ the number of rejected null hypotheses which are true nulls (false discoveries).
The two most commonly used measures of false discovery are the
Family-Wise Error Rate (FWER) and False Discovery Rate (FDR), defined
as follows.
\begin{eqnarray}
\begin{array}{rcl}
FWER & \equiv & P(V>0)\\
FDR & \equiv & E(\frac{V}{R}|R>0)P(R>0).
\end{array}
\end{eqnarray}
FWER is the probability
of making one or more false discoveries, while FDR is the expected
percentage of false discoveries among the total number of discoveries.

Given a measure of false discovery, one needs to develop statistical testing approaches that guarantee control over this measure. The simplest approaches typically assume nothing about the statistical hypothesis testing setup and guarantee control in all possible scenarios. Such
is the Bonferroni correction for controlling FWER.
The FDR control procedure developed
by~\cite{fdr}
assumes the test statistics are independent or
that they have a positive regression dependency~\citep{false_discovery_dependency}.
 A modified, more
conservative  version of
this procedure covers all other forms of dependency~\citep{false_discovery_dependency}.

An interesting situation arises when hypotheses arrive sequentially
in a stream, and the testing procedure should determine whether to
accept or reject each of them without waiting for the end of stream.
Organizing hypotheses sequentially may be a choice based on scientific considerations~\citep{alphainvesting}, but it may also be an external constraint, whereby the testing of each hypothesis is performed as it arrives and a conclusion is taken.

Such a constraint occurs when attempting to apply procedures for
controlling false discovery on a public database. Such databases,
accessed or distributed through the Internet, are
used by research groups worldwide, and are becoming increasingly
pervasive. Examples include Stanford university's HIVdb~\citep{HIVdb}
which serves the community of anti-HIV treatment researches, WTCCC's
large-scale data for whole-genome association studies~\citep{WTCCC} which is
distributed to selected research groups,
and the NIH Influenza virus resource~\citep{YimingBao01152008}, often
used to identify complex dependencies in the virus genome.
 Some databases
are targeted for a specific type of hypotheses and contain built-in
facilities for executing tests,
e.g., ProMateus~\citep{promateus}, a collaborative web server for
automatic feature selection for protein binding-sites prediction.

The need to control
false discovery across multiple research groups without sacrificing
too much power is repeatedly acknowledged~\citep{bias,gene_combine,gene_false,ion}.
Since such a public database serves multiple researchers, working
 independently at different times, a procedure for controlling false
 discovery in this context must be sequential, without prior knowledge
 of even the number of hypotheses that is going to be tested.
Such a framework, termed "Quality Preserving Database" (QPD),
is discussed in~\cite{qpd2010}.
A QPD controls false discovery by prescribing the levels of the tests
performed by its users. To avoid power loss over time, the database
size is continually augmented, and the costs of obtaining additional
data are fairly distributed among database users based on their
 activity.
Thus a QPD is potentially handling an infinite series of tests.

The QPD implementation described in~\cite{qpd2010} uses a sequential
testing procedure known as Alpha Spending~\citep{alphainvesting},
which is closely related to Bonferroni  correction.
Alpha Spending uses an $\alpha$-wealth
function $W(j)$ that contains a pool of available level for testing. If the
$j$'th test is conducted at level $\alpha_j$, then this amount is
subtracted from the pool, $W(j)=W(j-1)-\alpha_j$. Initializing
$W(0)=\alpha$ and requiring $\forall j : W(j)\geq 0$, guarantees that
$\sum_j \alpha_j\leq \alpha$, and hence that $FWER$ is controlled at
level $\alpha$.

In the search for a sequential procedure for controlling $FDR$,
a modified measure has been employed, $mFDR_\eta$~\citep{alphainvesting}:
\begin{equation}
mFDR_\eta\equiv \frac{E(V)}{E(R)+\eta},
\end{equation}
where $\eta>0$ is some constant, typically chosen to be
$\eta=1-\alpha$.
A sequential procedure known as
Alpha Investing~\citep{alphainvesting} controls
$mFDR_\eta$ at level $\alpha$. Alpha Investing is similar to Alpha
Spending in its usage of an $\alpha$-wealth function $W(j)$, except in
Alpha Investing this function not only decreases, but sometimes
increases. Whenever a hypothesis is
rejected, a reward is added to $W(j)$, allowing higher levels and higher power for subsequent tests.

In this paper we
define a more general framework we term Generalized Alpha
Investing. This framework allows more control over the level of the
$j$'th test and the reward obtained if the hypothesis is rejected.
We show that there is a trade-off between the two,
i.e., if a higher level is chosen for the $j$'th test, then
the reward decreases, and vice versa. Therefore, our general framework
changes the meaning of the wealth $W(j)$ function. Instead of just
$\alpha$-wealth, $W(j)$ becomes an abstract potential function, and
the amounts withdrawn from it can be converted to various combinations
of level and reward, of which Alpha Investing and Alpha Spending are
special cases.

We define an optimality criterion which determines what is the optimal
usage of amounts withdrawn from $W(j)$, and show that reaching the
optimum is feasible in the case of uniformly most powerful tests with
continuous distribution functions, a single parameter,
and a simple null hypothesis.
We demonstrate in simulations that this optimized variant is significantly
more powerful than Alpha Investing for a range of realistic scenarios.

We also present another  special case of Generalized
Alpha Investing which is more easily combined with systems that rely
on Alpha Spending. We term this last variant Alpha Spending with Rewards.
We demonstrate how Alpha Spending with Rewards can
replace Alpha Spending in a QPD implementation.
This replaces $FWER$ with $mFDR_\eta$ as the
controlled measure, in exchange for reduced usage costs.
Note that Alpha Investing, as formulated in~\cite{alphainvesting},
cannot be combined with a QPD, while Alpha Spending with Rewards
smoothly integrates with it.
Thus our new implementation demonstrates both the advantage
of a QPD, as well as the
advantages of our Generalized Alpha Investing approach.

The rest of the paper is organized as follows.
Section~\ref{section_mfdr} provides a concise summary of the Alpha Investing
procedure and mFDR as described in~\cite{alphainvesting}.
Section~\ref{section_gai} describes our generalization of the Alpha
Investing procedure, its special cases and associated optimality results.
Section~\ref{section_qpd} formally defines the QPD
along the lines of~\cite{qpd2010} and demonstrates how it can be
combined with Alpha Spending with Rewards, to control $mFDR$ and
achieve the QPD goals while imposing lower costs on users.

\section{Alpha Investing and mFDR}
\label{section_mfdr}

\subsection{Notation and Assumptions}

Following a notation similar to~\cite{alphainvesting},
we consider the problem of testing $m$ null hypotheses, $H_1$,
$H_2$,\ldots,$H_m$.
A random variable $R_j\in\{0,1\}$  indicates whether $H_j$ was
rejected, and $R(m)=\sum_{j=1}^{j=m} R_j$ counts the total number of
rejections. Similarly, $V_j \in \{0,1\}$ indicates the case $H_j$ is true
  and is rejected (i.e., type-I error), and $V(m)=\sum_{j=1}^{j=m}V_j$
  tracks total number of type-I errors.

Let $\Theta_j$ be the parameter space assumed by the $j$'th test.
The null hypothesis $H_j$ is defined as a subset $H_j\subset \Theta_j$,
and the alternative hypothesis is $\Theta_j-H_j$. $\theta_j$ is the true
parameter value, thus $H_j$ is true if $\theta_j\in H_j$.
We denote by $\Theta$ and $\theta$
the combined parameter space and true parameter values of all $m$
tests, respectively. We use the notation $P_\theta(\cdot)$ and
$E_\theta(\cdot)$ to denote probability distribution and expectation
when assuming $\theta$.
We denote the level of the $j$'th test by $\alpha_j$.

In our analysis we require an additional, uncommon measure we
term ``Best Power''. Denoted by $\ppj$, we define it as
the maximal probability of
rejection for any parameter value in the alternative hypothesis. The following definition
formalizes it.

\begin{definition}
The {\em Best Power}, $\ppj$, of testing a null hypothesis $H_j$ is
\begin{equation}
\ppj=sup_{\theta_j \in \Theta_j - H_j} P_{\theta_j}(R_j=1).
\end{equation}
\label{def_best_power}
\end{definition}

In {\em Best Power} we measure the upper bound on the
power, which is often 1.
The simplest case in which $\ppj<1$, is when the alternative hypothesis is simple, e.g., in Neyman-Pearson type tests. In this case the best power is simply the power computed for the alternative hypothesis. Another case in which $\ppj<1$ is when the alternative hypothesis is limited to a finite range. For example, a situation where a z-test is performed to test the effect of a drug on some clinical measure (we assume the standard deviation of this measure is known from previous studies). Since clinical measures are confined to some finite, biologically plausible range, it follows that the alternative hypothesis must be limited in range.

Similarly to~\cite{alphainvesting}, the results in this paper
require the following assumption on the dependence between the tests
in the sequence.
\begin{assumption}
\begin{eqnarray}
\forall {\theta\in H_j} :
P_\theta(R_j|R_{j-1},R_{j-2},\ldots,R_1)\leq \alpha_j \\
\forall {\theta\notin H_j} :
P_\theta(R_j|R_{j-1},R_{j-2},\ldots,R_1)\leq \ppj.
\end{eqnarray}
\label{ass_ind}
\end{assumption}
Assumption~\ref{ass_ind} requires that the level and power traits of a
test hold even given the history of rejections.
This is somewhat weaker than requiring independence.

\subsection{Alpha Spending}

An Alpha Spending procedure is a simple procedure for multiple
hypothesis testing.
It maintains a pool of nonnegative $\alpha$-wealth $W(j)$, where $W(0)$ is the
initial pool and $W(j)$ is the remaining wealth after the $j$'th test.
The level allocated for the $j$'th test, $\alpha_j$, is subtracted
from the wealth, i.e., $W(j)=W(j-1)-\alpha_j$.

Since the wealth must
remain non negative, this procedure enforces the invariant
$\sum_{j=1}^m \alpha_j \leq W(0)$. Hence this procedure
guarantees a bound of $W(0)$ on the Family-Wise
Error Rate (FWER), the probability of falsely rejecting any null
hypothesis,
$\forall m: FWER(m)=P_\theta(V(m)>0)\leq W(0)$.

The major drawback of this and other Bonferroni based procedures is
that the FWER measure is often too conservative, hence aiming to
control it results in loss of power. In the following we discuss
less conservative measures and procedures for controlling them.

\subsection{FDR and mFDR}

In~\cite{fdr} the concept of False Discovery Rate is introduced, which may
be defined as follows:
$FDR(m)=P_\theta(R(m)>0)E_\theta(\frac{V(m)}{R(m)}|R(m)>0)$. Since it
is trivially true
that $FDR(m)\leq FWER(m)$, power gain is to be expected when aiming to
control $FDR$ instead of $FWER$. Furthermore, under the complete null
hypothesis assumption, i.e., when assuming all null hypotheses are true, it holds that
$FDR(m)=FWER(m)$. Controlling $FWER$ under this assumption is referred to as
controlling it in the weak sense.

Multiple variants of $FDR$ have been suggested in the
literature, such as $pFDR$ which drops the term
$P(R(m)>0)$~\citep{storey_fdr}, and various variants termed {\em Marginal False Discovery Rate} (mFDR),
based on the form $E(V(m))/E(R(m))$, with or without adding a constant to the denominator~\citep{fdr}.

We adopt the form and notation of~\cite{alphainvesting}, in defining $mFDR_\eta$ as follows.
\begin{equation}
mFDR_\eta(m)=\frac{E_\theta(V(m))}{E_\theta(R(m))+\eta}.
\end{equation}

Under the complete null hypothesis assumption, it is easy to show
that $mFDR_{1-\alpha}(m)\leq \alpha$ implies
$E_\theta(V(m))\leq \alpha$. Since
$FWER(m)\leq E_\theta(V(m))$, then
similar to the standard $FDR$, controlling $mFDR_\eta$ with $\eta$ set
  to $1-\alpha$ implies control over $FWER$ in the weak sense.

In the rest of this paper we only refer to control in the strong
sense, i.e., a procedure that controls $mFDR_\eta(m)$ at level
$\alpha$ will guarantee it is bounded by $\alpha$ for any value of the
unknown parameters $\theta\in\Theta$.

\subsection{Alpha Investing}
Alpha Investing~\citep{alphainvesting} is a procedure that
controls $mFDR_\eta$ at level $\alpha$ for any given choice of $\eta$ and $\alpha$.
It is defined as follows.

Similarly to Alpha Spending, the procedure uses a pool of
$\alpha$-wealth, $W(j)$. While in Alpha Spending we only subtract
amounts from this wealth until it is depleted, here we sometimes put
back wealth to the pool, hence the increased power.

A deterministic function $\II_{W(0)}$, known as the alpha investing rule,
determines the level of the next test as a function of the current
rejection history and the initial wealth.
\begin{equation}
\alpha_j=\II_{W(0)}(\{R_1,R_2,\ldots,R_{j-1}\}).
\label{equ_ai_rule}
\end{equation}

The $\alpha$-wealth is initialized and updated as follows.
\begin{eqnarray}
\begin{array}{rcl}
W(0) & = & \alpha\eta\\
W(j) & = & W(j-1)-(1-R_j)\frac{\alpha_j}{1-\alpha_j}+R_j \omega,
\end{array}
\label{equ_ai_update}
\end{eqnarray}
i.e., if $H_j$ is accepted, then the amount
$\frac{\alpha_j}{1-\alpha_j}$ is subtracted from the wealth. If it is
rejected, a reward $\omega$ is added.

The following definition summarizes the Alpha Investing Procedure
\begin{definition}
An {\em Alpha Investing Procedure} is a sequential hypothesis testing
procedure such that the level of the $j$'th test is determined
by Alpha Investing rule~\ref{equ_ai_rule}, the $\alpha$-wealth is
initialized and updated by equations~\ref{equ_ai_update}, $\forall j : W(j)\geq 0$
and $\omega\leq \alpha$.
\label{def_ai}
\end{definition}

Note that since $\omega$ plays no additional role in this system other
than the reward, it usually simply set to the maximal allowed value
$\omega=\alpha$.
In \cite{alphainvesting} it is proven that under
Assumption~\ref{ass_ind},
an Alpha Investing procedure indeed controls $mFDR_\eta$ at level
$\alpha$.

\section{Generalized Alpha Investing Procedure}
\label{section_gai}

\subsection{Definition}
We now present
a more general definition of Alpha-Investing-like procedures.
We show that the original Alpha Investing procedure is a special
case of it, as well as Alpha Spending. We
show further an additional special case which is better suited than
Alpha Investing for performing Neyman-Pearson tests.

Our generalization starts with a generalized investing rule $\JJ_{W(0)}$, which now
separately determines three quantities:
The level of the test, $\alpha_j$,
the amount $\ccj$ subtracted from the wealth,
and the reward $\ddj$ received upon rejection.
\begin{equation}
(\ccj,\ddj,\alpha_j)=\JJ_{W(0)}(\{R_1,R_2,\ldots,R_{j-1}\}).
\label{equ_gai_rule}
\end{equation}

We no longer refer to $W(j)$ as $\alpha$-wealth, since its
relation to the level of the test is quite indirect. We
refer to it more generally as a potential function. Its initialization
and update rule are now:
\begin{eqnarray}
\begin{array}{rcl}
W(0) & = & \alpha \eta \\
W(j) & = & W(j-1)-\ccj+R_j \ddj.
\end{array}
\label{equ_gai_update}
\end{eqnarray}

Note one subtle point: In the original Alpha Investing (equation~\ref{equ_ai_update}),
the amount $\frac{\alpha_j}{1-\alpha_j}$ is subtracted from the
wealth only if we fail to reject $H_j$. Here the amount $\ccj$ is
subtracted unconditionally. This change is
introduced to simplify the analysis below.

\begin{definition}
A {\em Generalized Alpha Investing Procedure} is a sequential hypothesis testing
procedure such that $\alpha_j$, $\ccj$ and $\ddj$ are determined
by rule~\ref{equ_gai_rule}, the potential function is
initialized and updated by equations~\ref{equ_gai_update},
$\forall j : W(j)\geq 0$, and the following inequality holds:
\begin{equation}
\forall j : 0\leq \ddj \leq
min(\frac{\ccj}{\ppj}+\alpha,\frac{\ccj}{\alpha_j}+\alpha-1),
\label{equ_gai_max_reward}
\end{equation}
where $\ppj$ is the {\em best power} of the $j$'th test.
\label{def_genai}
\end{definition}

Before analyzing this framework, let us demonstrate that the original
Alpha Investing procedure is a special case of it.
\begin{proposition}
Alpha Investing procedure is a special case of the Generalized Alpha
Investing procedure.
\end{proposition}
\begin{proof}
In the generalized framework,
let us choose $\ccj=\frac{\alpha_j}{1-\alpha_j}$ and
$\ddj=\ccj+\omega$, where $\omega$ is some constant such that $0\leq \omega\leq \alpha$.
Substituting in the update rule~\ref{equ_gai_update} we get
$W(j)  =  W(j-1)-\frac{\alpha_j}{1-\alpha_j}+R_j
     (\frac{\alpha_j}{1-\alpha_j}+\omega)
      =  W(j-1)-(1-R_j)\frac{\alpha_j}{1-\alpha_j}+R_j \omega$,
which is the Alpha Investing update rule.
It only remains to show that the constraint~\ref{equ_gai_max_reward}
is satisfied.
This holds since $\omega\leq \alpha$ and $\ppj\leq 1$.\qed
\end{proof}

We now proceed to show that the Generalized Alpha
Investing Procedure indeed controls $mFDR_\eta$. We follow the same
main ideas as the proof for alpha-investing in~\cite{alphainvesting}.

\begin{theorem}
Given Assumption~\ref{ass_ind}, a Generalized Alpha Investing Procedure
controls $mFDR_\eta$ at level $\alpha$.
\label{theorem_gai_controls_mfdr}
\end{theorem}
\begin{proof}

Requiring $mFDR_\eta(m)\leq \alpha$ is equivalent to requiring:
$E_\theta(\alpha R(m)-V(m))+\alpha \eta \geq 0$. Since the procedure
uses a non-negative potential function $W(j)$, it is sufficient to
require $E_\theta(\alpha R(m)-V(m)-W(m))+\alpha \eta \geq 0$.

Let $A(j)\equiv \alpha R(j)-V(j)+\alpha \eta - W(j)$.
Lemma~\ref{lem_martingale} proves $A_j$ is a sub-martingale with
respect to $R_j$.
From the sub-martingale properties it follows that $E_\theta(A(j))\geq A(0)$.
Since by definition $V(0)=R(0)=0$ and $W(0)=\alpha\eta$, we get $A(0)=0$.
Hence $E_\theta(A(j))\geq 0$.\qed
\end{proof}

\subsection{The Generalized Alpha Investing Trade-off Function and Wealth
Management}

In the Alpha Spending procedure, whenever a test is performed at level
$\alpha_j$, the same amount $\alpha_j$ is subtracted from the
$\alpha$-wealth, $W(j)$. In the Alpha Investing procedure, for a test at
level $\alpha_j$ we need to subtract a larger amount of
$\frac{\alpha_j}{1-\alpha_j}$, as prescribed by equation~\ref{equ_ai_update}.
Both procedures do not specify a
particular way for managing the $\alpha$-wealth, but leave it to the
user to define rule~\ref{equ_ai_rule}.

Generalized Alpha Investing procedure offers even more freedom.
It requires the user to choose separately the level
of the test $\alpha_j$, the amount $\ccj$ subtracted from the potential function $W(j)$, and
the reward $\ddj$. Any choice is valid as long as the potential
function does not become negative, $\ccj\leq W(j-1)$, and the reward
satisfies constraint~\ref{equ_gai_max_reward}.

In order to study the relation between the three user defined
parameters $\alpha_j$, $\ccj$, and $\ddj$, let us fix $\ccj$ on a
given amount, and set the reward $\ddj$ to the maximal value allowed
by constraint~\ref{equ_gai_max_reward}. According to this constraint,
$\ddj$ now becomes a function of $\alpha_j$ (note that $\ppj$ used in
the constraint is also a function of $\alpha_j$).

\begin{figure}
\centering
\makebox{\includegraphics[scale=0.4]{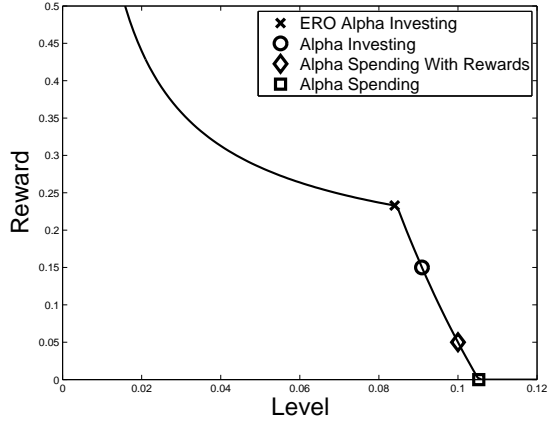}}
\caption{\label{fig_tradeoff}Trade-off between level and reward in Generalized Alpha Investing}
\end{figure}

Figure~\ref{fig_tradeoff} shows $\ddj$ as a function of $\alpha_j$,
for a test about the mean $\mu$ of a normal distribution with unknown
variance. We used a $t$-statistic with $9$ degrees of freedom,
$H_0:\mu=0$, $H_1:\mu=0.5$, $\alpha=0.05$, and $\ccj=0.1$.
The shape of the line reflects the two parts of
constraints~\ref{equ_gai_max_reward}:
$\frac{\ccj}{\ppj}+\alpha$ and $\frac{\ccj}{\alpha_j}+\alpha-1$,
where the convex part of the line left of the 'knee' is the latter and
the right part is the former.
It shows
the trade-off the Generalized Alpha Investing offers between the level
of the test $\alpha_j$ and the reward $\ddj$.
At one extreme, setting the maximal possible level
$\alpha_j=\frac{\ccj}{1-\alpha}$ results in no reward, $\ddj=0$. At
the other extreme, the reward tends to infinity, $lim_{\alpha_j\to 0} \ddj=\infty$.

In the following subsections we explore various options
of trade-off between $\alpha_j$ and $\ddj$, as shown in
figure~\ref{fig_tradeoff}.
Note that we do not deal theoretically with the question of how to choose $\ccj$.
As in Alpha Investing and Alpha Spending, management of the $W(j)$
pool is left for user discretion. In the simulation results of Subsection~\ref{subsec_gai_simulation} and Section~\ref{sec_qpd_simulation} we empirically explore some natural candidates.

Alpha Investing is a special case of Generalized Alpha
Investing, where $\alpha_j=\frac{\ccj}{1+\ccj}$. The corresponding
point is marked by a circle in the figure. This seems like an arbitrary point, but
Subsection~\ref{subsec_gai_opt} shows it is
not.

Subsection~\ref{subsec_gai_as} shows that the extreme point of
$\alpha_j=\frac{\ccj}{1-\alpha}$, is in fact Alpha Spending. It
further shows another special case we term Alpha Spending with Rewards
which has some practical benefits. The rhombus marked in the figure is
for the case $k=1$, where $k$ is a configurable parameter of the Alpha
Spending with Rewards method, as defined in the next subsection. Other
values of $k$ may reach other points in the figure.

Subsection~\ref{subsec_gai_opt} shows that, under certain assumptions, the 'knee' noticeable in
the figure
is the optimal point for a
criterion we term Expected Reward Optimal (\ero).
We call this special case ERO Alpha Investing.

\subsection{Alpha Spending with Rewards}
\label{subsec_gai_as}

In Generalized Alpha Investing the level of the test $\alpha_j$
may be chosen independently of choosing
the amount $\ccj$ subtracted from the potential function.
A special case, therefore, is when choosing
$\ccj=k \alpha_j$, for some constant $k$.
This special case is of interest, since it can be viewed as direct
extension of Alpha Spending: By scaling the units of the potential
function $W(j)$, we can restore its original meaning of
$\alpha$-wealth, and have the level of the test $\alpha_j$ be directly
subtracted from it.  This is in contrast to Alpha Investing, where
$\alpha_j$ is non-linearly dependent on $\ccj$. We term this special
case ``Alpha Spending with Rewards'',
defined more formally next.

We define an Alpha Spending with Rewards rule, that
determines the level of the $j$'th test and the reward in case $H_j$
is rejected, based on $W(0)$ and the history of rejections.
\begin{equation}
(\ddj,\alpha_j)=\JJ_{W(0)}(\{R_1,R_2,\ldots,R_{j-1}\}).
\label{equ_asr_rule}
\end{equation}

The wealth is initialized and updated as follows.
\begin{eqnarray}
\begin{array}{rcl}
W(0) & = & \frac{\alpha \eta}{k} \\
W(j) & = & W(j-1)-\alpha_j+R_j \ddj.
\end{array}
\label{equ_asr_update}
\end{eqnarray}

\begin{definition}

An {\em Alpha Spending with Rewards} procedure
is a sequential hypothesis testing
procedure such that $\alpha_j$ and $\ddj$ are determined
by rule~\ref{equ_asr_rule}, the $\alpha$-wealth is
initialized and updated by equations~\ref{equ_asr_update},
$\forall j : W(j)\geq 0$, and the following inequality holds:
\begin{equation}
\ddj \leq min(\frac{\alpha_j}{\ppj}+\frac{\alpha}{k},1-\frac{1-\alpha}{k}),
\end{equation}
where $\ppj$ the {\em best power} of the $j$'th test.
\label{def_asr}
\end{definition}

We now proceed to show Alpha Spending with Rewards is indeed a special
case of the Generalized Alpha Investing Procedure, hence it controls
$mFDR_\eta$ at level $\alpha$.

\begin{proposition}
Alpha Spending with Rewards is a special case of the Generalized Alpha
Investing procedure.
\label{prop_asr_is_gai}
\end{proposition}
\begin{proof}
Choosing $W'(j)=k W(j)$, $\ccj=k \alpha_j$ and $\ddj'=k \ddj$
we get a Generalized Alpha Investing procedure with potential function
$W'(j)$, costs $\ccj$ and rewards $\ddj'$.\qed
\end{proof}
\begin{corollary}
Given Assumption~\ref{ass_ind}, Alpha Spending with Rewards controls
$mFDR_\eta$ at level $\alpha$.
\end{corollary}

\begin{proposition}
Alpha Spending is a special case of Alpha Spending with rewards, and
hence of the Generalized Alpha Investing Procedure.
\end{proposition}
\begin{proof}
In an Alpha Spending with Rewards procedure,
choosing $k=1-\alpha$, we get $\ddj=0$,
and $W(j)=W(j-1)-\alpha_j$. Hence this is equivalent to Alpha Spending
with initial wealth $W(0)=\frac{\alpha\eta}{1-\alpha}$.\qed
\end{proof}

Choosing $k=1-\alpha$ is the minimal value for $k$ that has
nonnegative rewards. Increasing $k$ increases the rewards, at the
expense of reducing $W(0)$.

Alpha Spending with Rewards can be more easily combined
 with systems based on Alpha
Spending, since its equation for updating the $\alpha$-wealth is
 identical to that of Alpha Spending, except for the initialization
 and the rewards. This point is demonstrated in following sections with the QPD.

\subsection{Optimizing the Generalized Alpha Investing Procedure}
\label{subsec_gai_opt}

In this subsection we find the optimal point of
 trade-off between $\alpha_j$ and $\ddj$, when given a potential
 investment amount $\ccj$.

Our criterion of optimality is $E_\theta(R_j) \ddj$, the expected reward
of the current test.
This is a natural optimization criterion, since aiming to increase our
expected reward is related to the power of the current test as well as
to our potential to perform additional tests. Note that as $\alpha_j$
increases, $E_\theta(R_j)$ increases but $\ddj$ decreases because of
constraint~\ref{equ_gai_max_reward}. Therefore it is not trivial to
maximize their product, $E_\theta(R_j)\ddj$.

Since $\theta$ is an unknown parameter, we seek a procedure for
which optimality is guaranteed for all $\theta$, as in the
following definition.

\begin{definition}
A Generalized Alpha Investing Procedure is {\em expected reward
optimal} (\ero) if at any point in time $j$, the procedure chooses
$\alpha_j$ and $\ddj$ such that for any other values
$\alpha_j'$ and $\ddj'$ and for any $\theta\in \Theta$ it holds that
$E_\theta(R_j)\ddj\geq E_\theta(R_j')\ddj'$, where $R_j'$ is the
random variable corresponding to $R_j$ for a test executed at level $\alpha_j'$.
\label{def_expected_reward_optimal}
\end{definition}

An \ero~procedure can be found under some assumptions on the relation between level and power.
This occurs when choosing the $\alpha_j$ and $\ddj$ at the intersection between the two parts of constraint~\ref{equ_gai_max_reward}, as shown in figure~\ref{fig_tradeoff}. We term this choice \ero~Alpha investing:

\begin{definition}
An \ero~Alpha Investing procedure is a Generalized Alpha Investing procedure that chooses
$\alpha_j$ as the solution of
$\frac{\ccj}{\ppj}=\frac{\ccj}{\alpha_j}-1$
and
$\ddj=min(\frac{\ccj}{\ppj}+\alpha,\frac{\ccj}{\alpha_j}+\alpha-1)$.
\label{def_ero_alpha_investing}
\end{definition}

We now proceed to show that this is indeed an \ero~procedure under certain assumptions.
In the appendix we provide a general Theorem~\ref{theorem_ero_full}, which fully lists the set of required assumptions, and is rigorously proven. However, since these assumptions and proofs involve many technicalities, we provide here a simpler theorem, that captures some practical cases.

\begin{theorem}
For a series of
uniformly most powerful tests with continuous distribution functions, a single parameter $\theta$, and a simple null hypothesis $H_0:\theta=\theta_0$,
\ero~Alpha Investing is an \ero~procedure.
\label{theorem_gai_is_opt}
\end{theorem}
\begin{proof}
Proven by Theorem~\ref{theorem_ero_full}, and Lemma~\ref{lem_ump_satisfies}, given in the appendix.
\end{proof}

One simple case covered by Theorem~\ref{theorem_gai_is_opt} is a series of Neyman-Pearson tests, where both null and alternative hypotheses are simple. A more general case is when the tests are uniformly most powerful, and when the null hypothesis remains simple. Such a case may occur, for example, when performing a z-test with null hypothesis $H_0:\theta=0$, and alternative $H_1:\theta > 0$ (i.e., when it is known that $\theta$ is non-negative).

Now that we've established an \ero~procedure exists, we proceed to show that for unbounded alternative hypothesis, this procedure in fact reduces to the original Alpha Investing.

\begin{corollary}
When the \ero~Alpha Investing is an \ero~procedure, the original Alpha Investing is also an \ero~procedure if and only if
$\forall_j: \ppj=1$.
\label{cor_reduce}
\end{corollary}
\begin{proof}
Setting $\ppj=1$ in Definition~\ref{def_ero_alpha_investing} yields
$\alpha_j=\frac{\ccj}{\ccj+1}$, which is the level used by Alpha
Investing for a given cost $\ccj$. Any other value of $\ppj$ clearly
yields different levels.
\end{proof}

To summarize this subsection, we have shown that in some cases an \ero~procedure can be found. In these cases, if the alternative hypothesis stretches to infinity, then the original Alpha Investing is identical to \ero~Alpha Investing.
However, if the alternative hypothesis is simple or limited in range, then \ero~Alpha Investing will have larger expected rewards than Alpha Investing, and therefore we expect it to be more powerful.

The next subsection explores this expected theoretical gain in simulations.

\subsection{Simulation Results}
\label{subsec_gai_simulation}

\begin{table}
\caption{\label{table_constant}Simulation results for the 'constant'
allocation scheme. Best results, marked in bold, were found to be
significantly better using paired t-test (see text for details).}
\centering
\fbox{%
\begin{tabular}{p{6cm} | p{1cm} | p{1cm} | p{1cm} | p{1cm} | p{1cm}}
\em Procedure       & \em Tests        & \em True rejects & \em False rejects & \em mFDR \\
\hline
Alpha Spending &       10.000 &        0.283 &        0.041 &        0.032 \\
Alpha Investing &       15.908 &        0.435 &        0.064 &        0.044 \\
Alpha Spending with Rewards (k=1) &       15.188 &        0.416 &        0.061 &        0.043 \\
Alpha Spending with Rewards (k=1.1) &       17.888 &        0.476 &        0.066 &        0.044 \\
\ero~Alpha Investing &      {\bf 18.355} &      {\bf  0.504} &        0.073 &        0.048 \\
\end{tabular}}
\end{table}

\begin{table}
  \caption{\label{table_relative}Simulation results for the 'relative' allocation scheme. Best
  results, marked in bold, were found to be significantly better using
  paired t-test (see text for details).}
  \centering
\fbox{%
\begin{tabular}{p{6cm} | p{1cm} | p{1cm} | p{1cm} | p{1cm} | p{1cm}}
\em Procedure       & \em Tests        & \em True rejects & \em False rejects & \em mFDR \\
\hline
Alpha Spending &       66.000 &        0.559 &        0.043 &        0.028 \\
Alpha Investing &       81.529 &        0.868 &        0.086 &        0.045 \\
Alpha Spending with Rewards (k=1) &       81.202 &        0.856 &        0.083 &        0.044 \\
Alpha Spending with Rewards (k=1.1) &       81.631 &        0.849 &        0.084 &        0.044 \\
\ero~Alpha Investing &       {\bf 82.626} &       {\bf 0.905} &        0.093 &        0.048 \\
\end{tabular}}
\end{table}

\begin{table}
  \caption{\label{table_rel200}Simulation results for the 'relative-200' allocation scheme. Best
  results, marked in bold, were found to be significantly better using
  paired t-test (see text for details).}
  \centering
\fbox{%
\begin{tabular}{p{6cm} | p{1cm} | p{1cm} | p{1cm} | p{1cm} | p{1cm}}
\em Procedure & \em Tests & \em True rejects & \em False rejects & \em mFDR \\
\hline
Alpha Spending &      200 &        0.576 &        0.046 &        0.029 \\
Alpha Investing &      200 &        0.885 &        0.090 &        0.047 \\
Alpha Spending with Rewards (k=1) &      200 &        0.873 &        0.088 &        0.046 \\
Alpha Spending with Rewards (k=1.1) &      200 &        0.867 &        0.090 &        0.047 \\
\ero~Alpha Investing &      200 &       {\bf 0.931} &        0.101 &        0.051 \\
\end{tabular}}
\end{table}

We simulated various variants of Generalized Alpha Investing on a
sequence of independent z-tests. Each z-test was executed on a single sample
from $N(\mu,1)$, with $\mu=0$ chosen with probability $0.9$,
otherwise $\mu=2$. The null hypothesis states that $\mu=0$.

Five procedures which were discussed in this paper were tested:
Alpha Spending, Alpha Investing, \ero~Alpha Investing, and two variants of
Alpha Spending with Rewards, one with $k=1$ and another with
$k=1.1$. These two values were chosen since the former typically
produces
values of $\alpha_j$ larger than the \ero~Alpha Investing,
while the latter produces lower values.

Each variant was executed on the same random sequence of tests, after which we
recorded three measures: the number of tests that were executed until the potential function was
depleted, the number of false null hypotheses that were rejected (true rejects), and the
number of true null hypotheses rejected (false rejects).
We repeated this experiment $10,000$ times, which allowed us to obtain reliable estimates of the expectation of each
measure, and to perform paired t-tests to confirm significance of differences. We also estimated $mFDR_\eta=\frac{E(V)}{E(R)+\eta}$ by replacing the expectations with the observed means.

The $\alpha$-wealth pool (or potential, in the case of Generalized
Alpha Investing) was initialized and managed as follows.
We set $\alpha=0.05$ and $\eta=1-\alpha$. We used three schemes for
allocating $\ccj$. The 'constant' scheme defines
$\ccj=min(\frac{1}{10} W(0),W(j))$ and it continues to perform tests
until the potential function $W(j)$ is depleted. The 'relative' scheme
defines $\ccj=\frac{1}{10}W(j-1)$ and it continues until
$W(j)<\frac{1}{1000}W(0)$. The 'relative-200' scheme also defines
$\ccj=\frac{1}{10}W(j-1)$, only it stops not when $W(j)$ is depleted,
but exactly after 200 tests.

Tables~\ref{table_constant},~\ref{table_relative}, and~\ref{table_rel200}
depict the average results over the $10,000$ simulations.
The results demonstrate the superiority of the \ero~Alpha Investing in
all three cases.
In addition to the averages shown in the tables, we performed paired comparisons between \ero~Alpha Investing and the other procedures (recall that each procedure was executed on the same $10,000$, randomly drawn sequences of tests).
For example, in the 'relative' allocation scheme: \ero~Alpha Investing performed at least as many tests and at least as many true rejections as Alpha Investing in $99.9\%$ of the $10,000$ simulations.
A two-sided paired t-test confirms with p-value $2.8 \times 10^{-92}$ that it performed significantly more tests, and more true rejections with p-value $1.7 \times 10^{-36}$.
Across all simulation settings and all competing approaches, ERO Alpha Investing delivered more true rejections
in at least $99\%$ of the $10,000$ repetitions, and all two-sided p-values were no bigger than $10^{-17}$.

The 'relative-200' case depicted in
table~\ref{table_rel200}
stresses a
different aspect of the advantages of \ero~Alpha
Investing. In the 'constant' and 'relative' cases, demonstrated by tables~\ref{table_relative}
and~\ref{table_constant}, the \ero~Alpha Investing
manages to execute on average more tests than the other procedures,
thanks to its maximization of the expected rewards. The 'relative-200'
scheme neutralizes this advantage, since all procedures stop after 200
tests, discarding whatever amount they have left in the pool. However, as noted above, it still significantly outperforms all other procedures in terms of number of true rejections.
This is due to having it build and maintain a larger pool, on average, than the
other procedures. This allows it to make larger allocations to
individual tests, and make more discoveries.

Not surprisingly, the number of false rejections by ERO Alpha Investing is also significantly higher than other approaches in these examples. But since it is successful in attaining the desired $mFDR$ level, this does not diminish its advantage in delivering more true rejections.

The simulations above were performed with $10\%$ false null hypotheses, with an effect size of $2$. Both of these values are quite large, and help make the advantages of \ero~Alpha Investing more pronounced. We've repeated these simulations with other values as well. We have tried all the combinations of false null probability $p=0.1, 0.05, 0.02, 0.01$, and effect sizes $\delta=2, 1, 0.5, 0.2, 0.1$.
Naturally, in harder settings all procedures perform less rejections, and the differences between them diminish.
Still, compared with Alpha Investing, \ero~Alpha Investing performed significantly more tests in all of these cases (maximal p-value 0.006). Regarding true rejections, the advantage was significant only for $\delta=2$ when $p\geq 0.02$ and for $\delta=1$ when $p=0.1$.
Full results for the extreme case $\delta=0.1$ and $p=0.01$ are given in the Appendix.

\section{The Quality Preserving Database (QPD)}
\label{section_qpd}

Researchers usually access public databases freely, and
perform statistical tests unchecked. This of course leads to multiple
testing issues, since the users of this data may be spread around the
globe, and have no knowledge of each other. Lack of
control over type-I errors is widely acknowledged as a source of
misleading results~\citep{bias,gene_combine,gene_false,ion}.

The QPD addresses this issue by adding to the database
an additional management layer that
handles the statistical tests that are performed on it.
Instead of
allowing unrestrained access, each user wishing to perform a
statistical test must first describe the characteristics of this test
to the QPD's manager. The manager allocates the level at
which this test will be performed, so as to guarantee some measure of
the overall type-I error is controlled, e.g., FWER.

Since the total level allocation is limited, and since the number of
test requests the QPD should serve is unbounded, a stringent scheme of
continuously decreasing level allocations is required. This implies
gradual power loss. To avoid that, users of a QPD are required to
compensate for the usage of the level allocated for them.
This compensation is in the form of additional data samples the user
must provide.

Let us explain the intuition behind this with an example.  Imagine a
certain researcher $A$ would like to perform a single tail z-test
with known standard deviation $\sigma=10$, significance level $0.05$
and power $0.9$ at effect size $\Delta=1$.
Say the database contains $857$ samples which is exactly
enough to perform this test.
However, another researcher $B$ has just
completed an independent test on the same database with significance
level $0.001$ and power $0.44$. In order to control type-I errors, the database
manager tells researcher $A$ to reduce the significance
level of his test to $0.049$. Ordinarily this would require researcher
$A$ to reduce the power to some value less than $0.9$, which is undesirable.
Since researcher $A$'s problem is due to researcher $B$'s test,
researcher $B$ is asked to
compensate by adding six more samples to the database.
$863$ samples is exactly enough to allow
researcher $A$ to execute his test with significance level $0.049$ and
the original power $0.9$.

Of course requiring actual new samples from
a researcher is not always practical. Instead the required
number of samples may be translated to currency.
This amount will be
the cost researcher $B$ will be charged with for executing the test,
 and it will be used by the QPD's manager to obtain six new samples.

In the simplistic example just mentioned, the compensation requested
for the first statistical test allows for a single additional test.
The QPD's goal, however, is to implement a compensation system in
which an infinite series of tests can be executed while controlling
type-I errors.
Furthermore, we require that the  amount of compensation demanded for
a test will not be negatively influenced by the characteristics of other tests in
the series.
Thus we define two requirements that such a system should fulfill:
\textit{stability} and \textit{fairness}.
Stability means
that the cost charged for a particular test should always
be sufficient to compensate for all possible future tests. In other words, it should not adversely affect the costs of future tests.
The fairness property requires that users that request more difficult
tests (e.g., with higher power demands) will be assigned higher costs.
We define these requirements formally below.

Thus the QPD can be viewed as a mechanism
for distributing the costs of maintaining a public database
fairly between the consumers of this data. In any active research domain
new samples must be continually obtained or
false discoveries will start accumulating unchecked. Thus it is fair to
request the researchers to participate in financing this task,
an amount proportional to the volume of their activity, as does the
QPD. In subsequent sections we demonstrate that employing Generalized Alpha
Investing with a QPD reduces these costs to practically zero if
properly used.

\subsection{Formal Definition of the QPD}
\label{subsec_qpd}


A QPD~\citep{qpd2010} serves a series of test requests.
Each request includes the following information: The test statistic,
assumptions on the distribution of the data (e.g. that it is normally
distributed), the desired effect size and power requirements.
The request does {\em not} contain two details: the
required significance level and the number of samples which will be
used to calculate the test statistic. These two details
are managed by the QPD's manager. However, given all the other request
details the significance level becomes a function of the number of
samples. We term this function the ``level-sample'' function, formally
defined below.
\begin{definition}
Given a test request containing test statistic, data assumptions,
and the desired effect size and power requirements, the
{\em Level-sample} function, $L:\mathbb{N}\to \mathbb{R}$, is a function specifying the
feasible significance level given the number of samples.
\end{definition}

The level-sample function summarizes the test request for the sake of
determining the cost required for executing it. The costs assigned for
requests may vary between different requests, and also may vary with
time, but the allocation scheme must fulfill two properties: {\em fairness} and {\em stability}.
The fairness requirement essentially says that easier
tests (e.g., those that require less power),
should have lower costs. The stability requirement dictates that costs
will not diverge with time, i.e., that some constant bound may be
precomputed for each type of request. The two are formally defined below.

\begin{definition}
The costs assigned by the QPD satisfy the {\em fairness} requirement
if for any two requests such that one has level-sample function
$L_a(n)$ and the other $L_b(n)$, and $\forall n : L_a(n)<L_b(n)$,
then at any particular point in time the first request will be
assigned with no higher cost than the second request.
\label{def_fairness}
\end{definition}
\begin{definition}
The costs assigned by the QPD satisfy the {\em stability} requirement
if for any particular request $a$ there is some constant $c_a$ such
that the cost assigned to it will never exceed $c_a$.
\label{def_stability}
\end{definition}

\begin{definition}
A Quality Preserving Database (QPD) is a database with a management
layer that assigns costs in the form of additional data samples for
each test executed. This layer fulfills the following three properties:
(a) It can serve an infinite series of requests,
(b) It satisfies the fairness and stability requirements,
(c) It controls some measure of the overall type-I errors (e.g.,
  $FWER$) at some pre-configured level $\alpha$.
\label{qpd_essentials}
\end{definition}

Definition~\ref{qpd_essentials} can be more simply stated as follows: The QPD allows users to perform unlimited tests with whatever power they choose. Their requirements will be fulfilled at a cost that reflects the difficulty of the request. The cost required from one user will not be adversely affected by the activity of previous users (but in practice we typically observe a gradual cost decrease, as shown in Section~\ref{sec_qpd_simulation}).

In~\cite{qpd2010} two possible implementations of a QPD are
presented, termed 'persistent' and 'volatile'.
The persistent version uses Alpha Spending to control
FWER. We now describe it, and in the
next subsection show how Alpha Spending can be replaced by Alpha
Spending with Rewards. The volatile version will not be discussed in
this paper.

Let $W(j)$ be the $\alpha$-wealth remaining after the $j$`th
test.
$W(0)$ is the initial $\alpha$-wealth set by the QPD
manager, which guarantees $FWER\leq W(0)$.
Let $c_j$ be the cost assigned to the $j$'th test, i.e., $c_j$ is the
number of additional samples required prior to executing the $j$'th test.
Let $n_j$ be the number of samples the database has after the
$j$'th test, i.e., $n_j=n_0+\sum_{k=1}^{j} c_k$, where $n_0$ is the
initial number of samples.

The crux of the implementation is the way the level of the test
$\alpha_j$ and the cost assigned for executing it $c_j$ are calculated.
Given $c_j$, the level is computed according to the following equation:
\begin{equation}
\alpha_j=W(j-1)(1-q^{c_j}),
\label{equ_persistent_level}
\end{equation}
where $0<q<1$ is some pre-configured constant.
Equation~\ref{equ_persistent_level} states that the level allocated
for the $j$'th test is fraction $1-q^{c_j}$ of the available
$\alpha$-wealth, $W(j-1)$, thus the higher the cost $c_j$, the higher the
level that is allocated.

Since $W(j)=W(j-1)-\alpha_j$, we obtain $W(j)=W(j-1) q^{c_j}$. From
this it follows by induction that the remaining $\alpha$-wealth is
exponentially decaying with the number of samples: $W(j)=\alpha q^{n_j-n_0}$.

Let $L_j:\mathbb{N}\to \mathbb{R}$ be the level sample function
of the $j$'th test. The cost $c_j$ assigned for this test
is calculated to be the minimal value that satisfies equation~\ref{persistent}.
\begin{equation}
L_{j}(n_{j-1}+c_j)\leq W(j-1)(1-q^{c_j}).
\label{persistent}
\end{equation}
This guarantees
that the level allocation for this test $W(j-1)(1-q^{c_j})$ will be
exactly enough to execute it given the existing number of samples and
the additional ones that will be obtained $n_{j-1}+c_j$.

Let us denote this particular implementation based on Alpha Spending
by $\QPDAS(\alpha,q)$. The
following definition summarizes this approach.
\begin{definition}
$\QPDAS(\alpha,q)$ is a QPD that uses Alpha Spending, with $W(0)=\alpha$,
  costs calculated as the minimal values that satisfy
  equation~\ref{persistent} and level allocations according to
  equation~\ref{equ_persistent_level}.
\end{definition}

\begin{theorem}
$\QPDAS(\alpha,q)$ fulfills the three properties
of Definition~\ref{qpd_essentials}, where stability is only guaranteed for
  requests such that their level-sample
function $\LN$ satisfies $L(n)\leq b q^n$ for some $b$.
\label{theorem_qpd_is_qpd}
\end{theorem}

Proof of Theorem~\ref{theorem_qpd_is_qpd} is given in~\cite{qpd2010}.
Note it fulfills the stability requirement only for tests
whose level-sample function decays
exponentially. This is true for  a wide range of commonly used
statistical tests, among them Neyman-Pearson tests, tests about the
mean of a normal distribution with known variance, and single-tail,
uniformly most powerful tests  with one unknown
parameter~\citep{qpd2010}. Simulations show that many more types of tests are
applicable in practice, e.g., tests based on an approximately normal
distribution such as a t-distribution. Finally, it is worth
emphasizing that even if some type of test does not satisfy stability,
it only means that the costs associated with such tests
may theoretically increase with time, but all the other properties of the
QPD remain. In particular, false discovery control is guaranteed for
all type of tests~\citep{qpd2010}.

\subsection{The Quality Preserving Database with Alpha Investing}
\label{subsec_qpdalpha}
We proceed to show how to combine the ideas of
the QPD and Generalized Alpha Investing.

Since the $\QPDAS(\alpha,q)$ uses the Alpha Spending procedure,
it would be simplest to replace the Alpha Spending procedure with the
Alpha Spending with Rewards variant of Generalized Alpha Investing, as
in Definition~\ref{def_asr}.
For simplicity of notation, we use $k=1$, though it will be clear
the results are valid for any $k$.

Let us denote by $\QPDASR(\alpha,\eta,q)$ that novel combination of
Alpha Investing with Rewards and the QPD.
We use the same equation~\ref{equ_persistent_level}
 for calculating the level $\alpha_j$
allocated for the $j$'th test,
and also calculate the costs in the same way using equation~\ref{persistent}.
However, the $\alpha$-wealth is initialized and updated according
to equations~\ref{equ_asr_update} of
Alpha Spending with Rewards. Thus,
we have that $\QPDASR(\alpha,\eta,q)$
controls $mFDR_{\eta}$ at level $\alpha$. Recall that
choosing $\eta=1-\alpha$ implies weak control over $FWER$ at level
$\alpha$.

The following definition and theorem formalize the above description,
and prove this is indeed a QPD.

\begin{definition}
$\QPDASR(\alpha,\eta,q)$ is a QPD that uses Alpha Spending with Rewards,
with $\alpha$-wealth initialized and updated according
to equation~\ref{equ_asr_update},
  costs calculated as the minimal values that satisfy
  equation~\ref{persistent} and level allocations according to
  equation~\ref{equ_persistent_level}.
\end{definition}

\begin{theorem}
$\QPDASR(\alpha,\eta,q)$ fulfills the three properties
of Definition~\ref{qpd_essentials}, where stability is only guaranteed for
  requests such that their level-sample
function $\LN$ satisfies $L(n)\leq b q^n$ for some $b$,
and type-I error control requires Assumption~\ref{ass_ind}.
\label{theorem_qpdr_is_qpd}
\end{theorem}
\begin{proof}
From Lemma~\ref{lem_qpdr_superior} it follows that
$\QPDASR(\alpha,\eta,q)$ always maintains a positive
$\alpha$-wealth. It is clear that given a positive $\alpha$-wealth,
it is always possible to satisfy
equation~\ref{persistent} by choosing a large enough $c_j$.
Therefore, $\QPDASR$ can serve an infinite series of requests.

Fairness immediately follows from equation~\ref{persistent}.

Stability can be shown by repeating the proof of stability in~\cite{qpd2010}:
Let $\LN$ be a ``level-sample''
function of a request satisfying $\forall n: L(n)\leq bq^n$ and
$c^\star=log_q(\frac{\alpha\eta q^{-n_0}}{b+\alpha\eta q^{-n_0}})$
($c^\star$ is a positive number for $0<q<1$ and $b>0$).
The following series of inequalities again uses Lemma~\ref{lem_qpdr_superior}.
\begin{eqnarray}
\begin{array}{rcl}
L_j(n_{j-1}+c^\star)
&\leq &b q^{n_{j-1}+c^\star}
\;=\; b q^{n_{j-1}} \frac{\alpha\eta q^{-n_0}}{b+\alpha\eta q^{-n_0}}
\;=\;  \alpha\eta q^{n_{j-1}-n_0} \frac{b}{b+\alpha\eta q^{-n_0}}\\
&\leq&W(j-1) \frac{b}{b+\alpha\eta q^{-n_0}}
\;=\; W(j-1) (1-\frac{\alpha\eta q^{-n_0}}{b+\alpha\eta q^{-n_0}})
\;=\;  W(j-1) (1-q^{c^\star}).
\end{array}
\end{eqnarray}
Hence equation~\ref{persistent} is satisfied for $c^\star$, which
proves it is an upper bound on the
cost associated with this request.

The third property is fulfilled by controlling $mFDR_\eta$ at level
$\alpha$, which is implied by Proposition~\ref{prop_asr_is_gai},
Theorem~\ref{theorem_gai_controls_mfdr} and Assumption~\ref{ass_ind}.\qed
\end{proof}

Controlling a less conservative measure must result in some
gain. Usually, this gain is in the form of increased power. However in
the QPD scheme power is specified by the user's request and the level
allocation is carefully calculated to support this power specification exactly.
Hence here we expect the gain to be in the form of reduced costs.
Every time a hypothesis is rejected a large amount is added to the
potential function $W(j)$. This results in a sharp decrease in
the costs assigned by equation~\ref{persistent}.

This simple implementation is pessimistic: while it exploits the additional
potential received from rejected null hypotheses, it never
expects any of it, and always allocates level as if no more
$\alpha$-wealth will ever be obtained.

A more optimistic implementation may be employed to obtain even
smaller costs. Let us term this implementation $\QPDASROPT(\alpha,\eta,q)$.
In this implementation we
present the $\alpha$-wealth $W(j)$ as the sum of two pools of wealth:
$W(j)=A(j)+B(j)$. The first pool, $A(j)$, is the same
$\alpha$-wealth that a $\QPDAS(\alpha \eta,q)$ would have. The second pool, $B(j)$
is the additional wealth obtained by rejecting hypotheses.

When a test request is served its level allocation is composed of two
parts. The amount withdrawn from the $A(j)$ pool, is the same as in
the $\QPDAS$, i.e., $A(j-1)(1-q^c_j)$. As in the $\QPDAS$, this amount
depends on the cost $c_j$.

The $B(j)$ pool, on the other hand, is managed more openhandedly, knowing it
will be refilled with $\alpha$ every time a rejection occur.
At each step we estimate the probability of this happening based
on prior history
$p(j)=\frac{1}{j-1} \sum_{i=1}^{j-1} R_j$,
and withdraw the following amount for the $j$'th test:
$min(p(j) \alpha, B(j-1))$.
Albeit somewhat heuristic, this management of the $B(j)$ pool does
have the following property. If we may assume a real probability of rejection
$p$, then the expected amount withdrawn at the $j$'th step from $B(j)$
is $E(min(p(j)\alpha,B(j-1)))\leq E(p(j)\alpha)=p\alpha$. The expected
reward of the $j$'th test is
also $p\alpha$. Therefore the $B(j)$ pool is kept in a balanced
manner. Note that the amounts withdrawn from it are irrespective of
$c_j$, therefore can be thought of
as 'free of charge'.

The following equations formally define the level allocation of the
$j$'th test.
\begin{equation}
\begin{array}{rcl}
p(j) & = & \frac{1}{j-1} \sum_{i=1}^{j-1} R_i\\
\alpha_j & = & A(j-1)(1-q^c_j) +min(p(j)\alpha,B(j-1)).
\end{array}
\label{equ_qpdropt_level}
\end{equation}

As in the $\QPDAS$, the cost is calculated to be the minimal value such
that the level allocation is just enough to execute the test.
\begin{equation}
L_j(n_{j-1}+c_j)\leq \alpha_j.
\label{equ_qpdropt_cost}
\end{equation}

Now we can formally define the initialization and update rules of the
$\alpha$-wealth pools.
\begin{equation}
\begin{array}{rclrcl}
A(0) &=& \alpha\eta &
B(0)&=&0\\
A(j)&=&A(j-1)q^{c_j} &
B(j)&=&B(j-1)-min(p(j) \alpha, B(j-1))+\alpha R_j\\
W(j)&=&A(j)+B(j). &&&
\end{array}
\label{equ_qpdropt_update}
\end{equation}

The following definition summarizes the $\QPDASROPT(\alpha,\eta,q)$.
\begin{definition}
$\QPDASROPT(\alpha,\eta,q)$ is a QPD that uses Alpha Spending
with Rewards,
with $\alpha$-wealth initialized and updated according
to equations~\ref{equ_qpdropt_update},
costs calculated as the minimal values that satisfy
equation~\ref{equ_qpdropt_cost} and level allocations according to
equations~\ref{equ_qpdropt_level}.
\end{definition}

\begin{theorem}
$\QPDASROPT(\alpha,\eta,q)$ fulfills the three properties
of Definition~\ref{qpd_essentials}, where stability is only guaranteed for
  requests such that their level-sample
function $\LN$ satisfies $L(n)\leq b q^n$ for some $b$,
and type-I error control requires Assumption~\ref{ass_ind}.
\label{theorem_qpdropt_is_qpd}
\end{theorem}
\begin{proof}

As in the $\QPDAS$, it is easy to show by induction that
$A(j)=\alpha\eta q^{n_j-n_0}$. Since $B(j)\geq 0$, it follows that
$W(j)\geq \alpha\eta q^{n_j-n_0}$. Hence $\QPDASROPT$ maintains a
positive $\alpha$-wealth
which allows it to serve an infinite series of
requests.

Stability and fairness follow, using the same argument as in the proof
of Theorem~\ref{theorem_qpdr_is_qpd}.

From equations~\ref{equ_qpdropt_update} and~\ref{equ_qpdropt_level} we
may write $W(j)=W(j-1)-\alpha_j+R_j \alpha$ and $W(0)=\alpha\eta$.
Hence $\QPDASROPT$
indeed implements Alpha Spending with Rewards, and therefore, under
Assumption~\ref{ass_ind}, controls $mFDR_\eta$ at level $\alpha$.\qed
\end{proof}

\subsection{Practical Considerations of QPD Usage}

The QPD notion may raise concerns that it is unfeasible in practice. Some of these concerns have already been discussed in~\cite{qpd2010}, but we'd like to revisit this issue here again, and expand on it. In particular, we address in this subsection three key questions: Is the formalism of requiring users to specify power and effect-size of their tests too restraining? Can researchers be convinced to pay for performing statistical tests? Can an example be given of a particular research community, where a QPD can realistically be thought of as being accepted and used?


Regarding the first question, one might argue that most researchers cannot specify in advance either power or effect-size. They simply perform a statistical test, e.g., a t-test, and if the p-value is lower than some acceptable threshold, they declare finding a significant result.
The problem is that usually the significance threshold cannot be properly selected, since the researchers don't know how many other research groups are testing similar hypotheses on the same data. In cases where the threshold is corrected for multiple testing, the power of the tests may drop so low that no reasonable effect could be detected.

The QPD relieves the researcher from the need to worry about p-values and thresholds. These are protections against type-I errors, which are centrally managed by the QPD. Instead, the researchers should consider and balance quantities that are much more of interest to them: how much will it cost to perform a test, and how likely is the test to discover the effect size they hope or expect to discover.

Note that a researcher need not come up with power and effect-size requirement of her own in advance. Instead the QPD manager can offer her various trade-off options. For example, he may offer to perform a t-test with power 0.9 for effect-size 0.5 for free, or power 0.99 for effect-size 0.5 for the cost \$50. This gives the researcher an idea of what her test is capable of performing, and of course she is free to inquire what can be obtained for other costs, or how the power is measured for other effect-sizes.
This allows the researcher to put aside p-values and type-I errors, and focus on what she aims to find, and plan a test that is likely to find it.

The second issue is the requirement that researchers pay for performing
their statistical tests. The idea that scientific research requires funding --- for data collection, lab work, even paying for publication in some journals --- is well engrained in the scientific research community,
and we believe that viewing the statistical testing on public data as an additional aspect
that needs to be funded is not a major leap of faith. This is particularly
true since paying for QPD usage is typically (in some cases provably as we
show in \cite{qpd2010}) cheaper than the cost of collecting data
required for testing the researcher's hypothesis  independently. As we show
in Section~\ref{sec_qpd_simulation}, in realistic scenarios QPD usage may in fact be free for most users,
and cheap for others, especially with the enhancements we presented in this
paper.

Because QPD represents an efficient use of the data for scientific research, another more radical view on the cost issue could be that funding agencies (like NIH) would prefer to take on the cost of creating and managing a QPD themselves, rather than funding the data collection efforts of individual scientists. In this scenario, the scientists would submit proposals for QPD usage instead of funding requests as is the common practice today.
To support continuous research, the funding agency may supply a constant flow of samples into the database, or have access to "batches" of samples. It is straightforward but beyond the scope of this paper to design schemes that could manage tests while allowing for such modes of data collection.

One more potentially problematic aspect of QPD is the separation it enforces between the researcher and the data being used for testing, limiting opportunities for exploratory data analysis, data cleaning etc. These concerns can be ameliorated in two ways: if the QPD managers are widely accepted as capable and reliable to perform these tasks themselves, or if QPD managers completely trust researchers to make use of the database only as designated, in which case it can be shared with them.

We conclude this subsection with a detailed example where we think a QPD can be successfully implemented. Genome wide association studies (GWASs)
are a common practice for finding relations between genetic information (genotype) and traits of organisms (phenotype).
In a single GWAS an order of 1 million single nucleotide polymorphisms (SNPs) on the genome are tested for association with a particular phenotype, e.g., a certain disease.

Due to the vast amount of hypotheses tested, as well as other issues, such as noise in sampling and stratification, type-I errors are very abundant in GWASs.
Because of that, leading journals require major results found on one dataset, to be replicated on a different, independent dataset, prior to publication~\citep{kruglyak1}. This is required regardless of the strength of statistical evidence in the initial study.

Because of the costs and organizational difficulties associated with data collection for GWAS, a common practice is to form consortia that join forces to fund and perform multiple GWASs on one or more phenotypes. Examples include the Wellcome Trust Case Control Consortium~\citep{wellcome2}, the Alzheimer's Disease Genetics Consortium~\citep{naj1}, and many others. The sizes of these consortia vary greatly, from two research groups to entire research communities.

This is an ideal starting point for a QPD. We envision the following scenario. An entire  community of researchers of a particular disease, e.g., type 2 diabetes (T2D), will set up a  {\em QPD replication consortium} to play the role of the independent dataset for replication testing. They will collect an initial amount of, say 500 cases (i.e., people having the disease) and 500 controls (healthy individuals), with genomic data for each individual, following standard GWAS practice. Protocols can be developed which will control how the QPD size is increased with usage, following our schemes proposed above.

The data of the QPD will not be made available to researchers, who will instead perform their research on their own (or their GWAS consortium's) legacy data, as is the common practice today~\citep{hardy1}. The QPD will only serve for replication in case the researchers wish to declare findings. This is aligned both with the requirement for independent data for replication testing, as well as the requirement to carefully control all tests performed on the QPD.

Say a research group believes it has found a list of ten SNPs associated with T2D. This may be due to an independent GWAS the group performed, or by using other means, e.g., study of genetic pathways. The QPD will make an ideal candidate for replicating these results, and giving the final approval for their correctness. Its data is hidden, and therefore the results are truly independent. Only the ten suspected SNPs should be tested, so the cost will be low. Finally, the QPD's type-I error control will give a strong protection against false discovery across the community, while maintaining reasonable power.

\section{Simulation Results}
\label{sec_qpd_simulation}

We have simulated the three implementations of QPD:
The original implementation from~\cite{qpd2010} that uses Alpha Spending,
$\QPDAS(0.05,0.999)$,
our new implementation with Alpha Spending with Rewards, $\QPDASR(0.05,0.95,0.999)$,
and its optimistic variant,
$\QPDASROPT(0.05,0.95,0.999)$.
All three are serving the same sequence of requests:
$100$ requests for t-test with power
$0.95$ and effect size $0.1$, simulating independently distributed
statistics.
The real probability of a true null is $0.9$.
The initial number of samples $n_0=2000$.
For each implementation we performed $1000$ realizations of this simulation.

Figure~\ref{fig_compare} depicts the results. Figure~\ref{fig_compare_a} shows the mean cost of the $i$'th test for each implementation.
The gradual cost decrease of $\QPDAS$ was already mentioned and
discussed by~\cite{qpd2010}. This results from having the levels of
the tests decay faster than assumed as the number of samples
grows.
This decay can be thought of as a reward for waiting,
since users who delay their tests risk losing the novelty of
their hypotheses, but are compensated with lower costs.
As expected, since $\QPDASR$ and $\QPDASROPT$ control $mFDR_\eta$
instead of $FWER$, they can achieve the requested power using smaller
number of samples, hence resulting with lower costs.

Figure~\ref{fig_compare_b} shows the ratio between the costs required by $\QPDASR$ and by $\QPDAS$. It shows the mean, median, and the 2.5th and 97.5th percentiles of these ratios over the $1000$ realizations. Starting after $41$ tests the costs of $\QPDASR$ drop to less than one third, on average, than the costs of $\QPDAS$, and the 97.5th percentile is at $0.7$.

Figure~\ref{fig_compare_c} shows the ratio between the costs required by $\QPDASROPT$ and by $\QPDAS$.
The reduction in cost here is much more dramatic.
The $\QPDASROPT$ cost drops to $0$ immediately after the first
rejection, and from there on, it usually remains so. This is because the $B(j)$ pool
fills up with the first reward, and is then generously distributed to
subsequent tests. Those tests get so high a level, that no further
samples are needed in the database to support the required power. Even
though the $B(j)$ pool is used generously, the probability of it
refilling before running out is so high, that frequently it is enough
to supply level allocation for all the remaining tests.
All $1000$ realizations have a cost of $0$ from the 78th test onward with $\QPDASROPT$. The 97.5th percentile drops to $0$ from the 39th test onward.

Comparing the number of rejected hypotheses, all three implementations reject the exact same number of false null hypotheses, $9.68$ on average.
Recall that
a QPD, regardless of the implementation, serves requests that specify
a power requirement, hence all tests in this simulation were performed
with power $0.95$, which explains this identical result (identical realizations of sequences of tests were served by each of the QPD variants).
The differences between implementations are in the
number of samples present at the $j$'th test, and the level allocation
$\alpha_j$. $\QPDASR$ and $\QPDASROPT$, enjoying the rewards, work with
lower number of samples and larger level allocations. These two effects
cancel each other to produce the same power, but produce more type-I
errors: $\QPDAS$ had $0.029$ type-I errors on average, $\QPDASR$ had $0.063$, and $\QPDASROPT$ had $0.116$. This is expected since $\QPDAS$ controls the $FWER$, while the other two control $mFDR$.

\begin{figure}[htb]
\centering
\subfigure[Mean costs]{
\makebox{\includegraphics[scale=0.3]{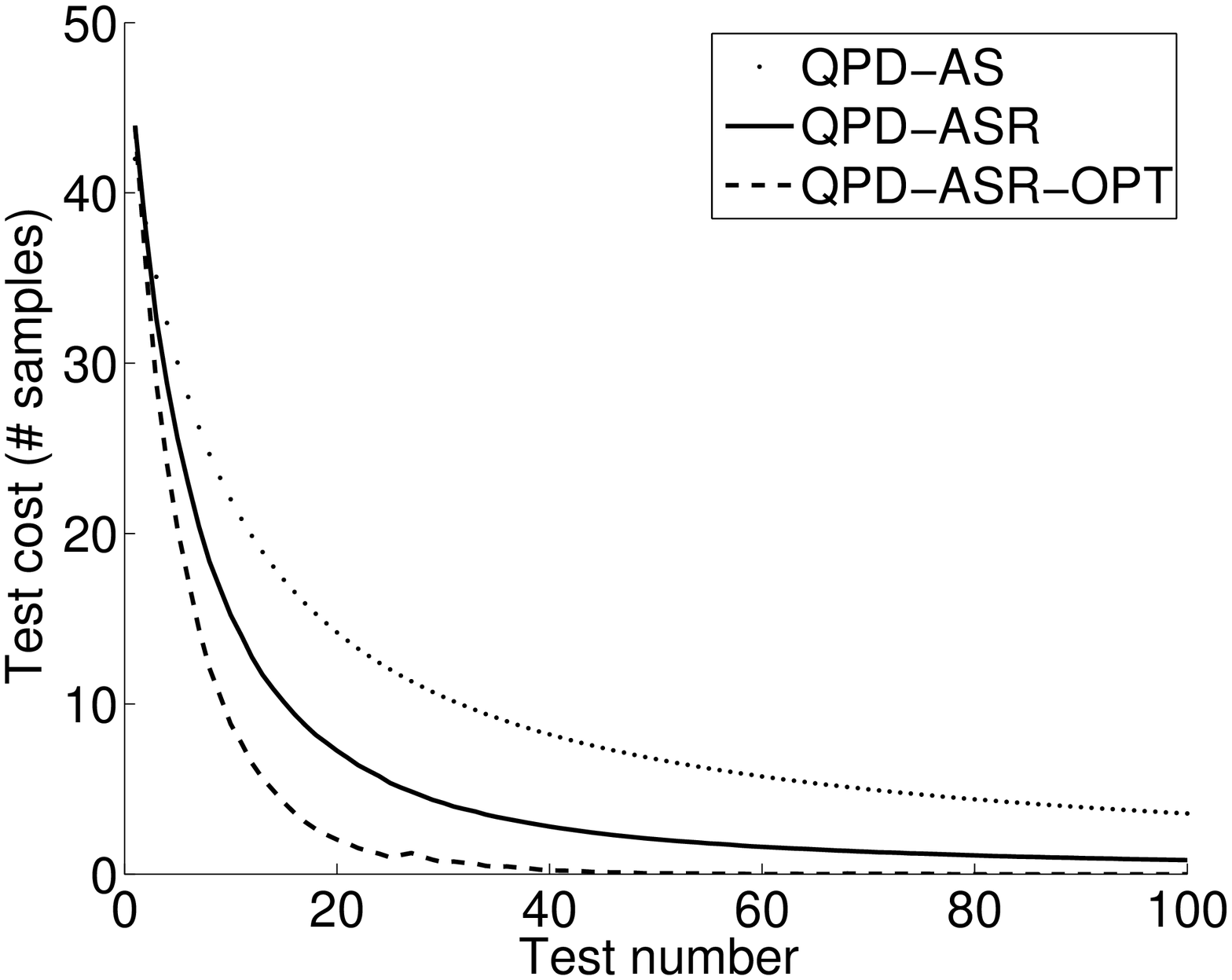}}
\label{fig_compare_a} } \subfigure[QPD-ASR vs. QPD-AS]{
\makebox{\includegraphics[scale=0.3]{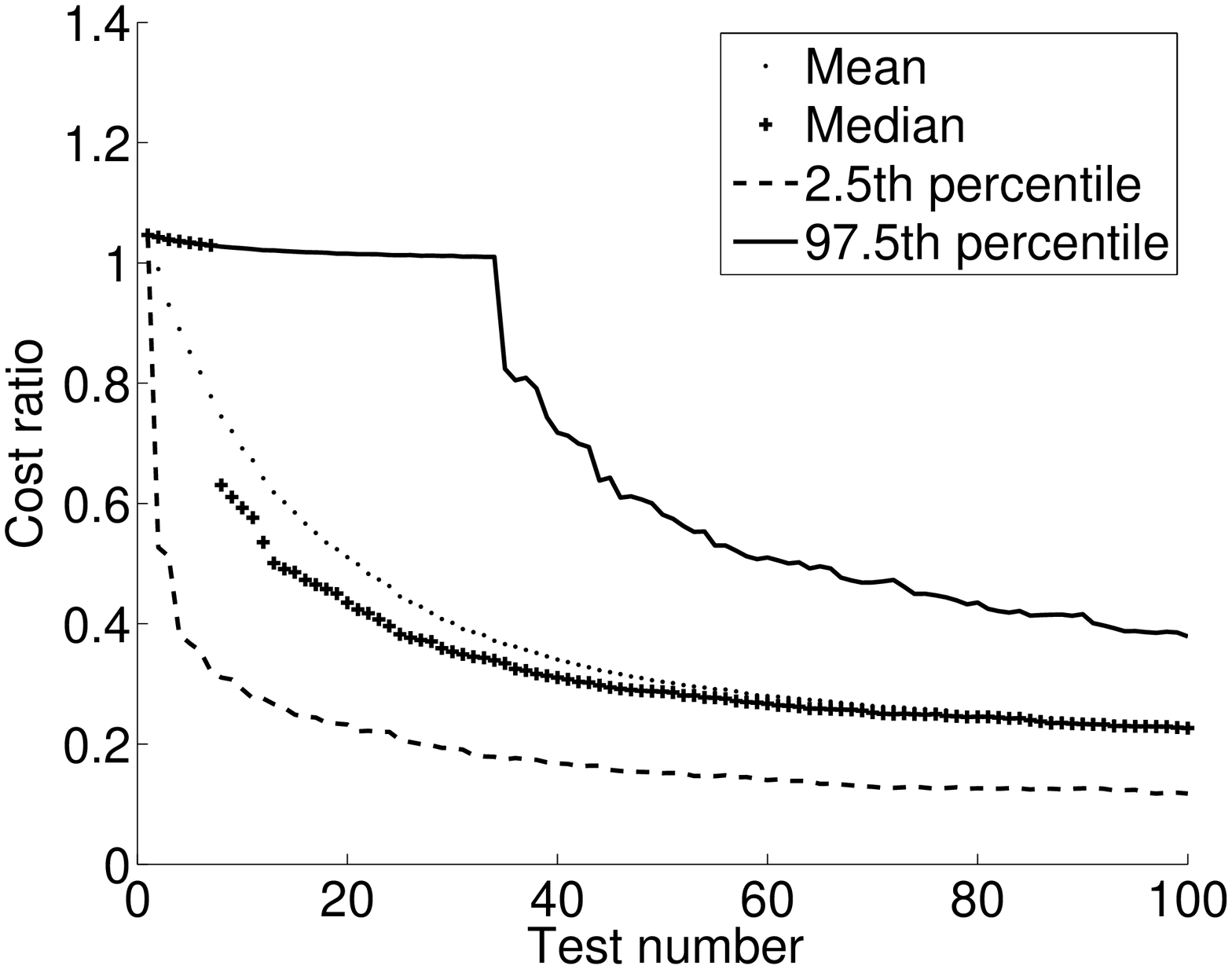}}
\label{fig_compare_b} } \subfigure[QPD-ASR-OPT vs. QPD-AS]{
\makebox{\includegraphics[scale=0.3]{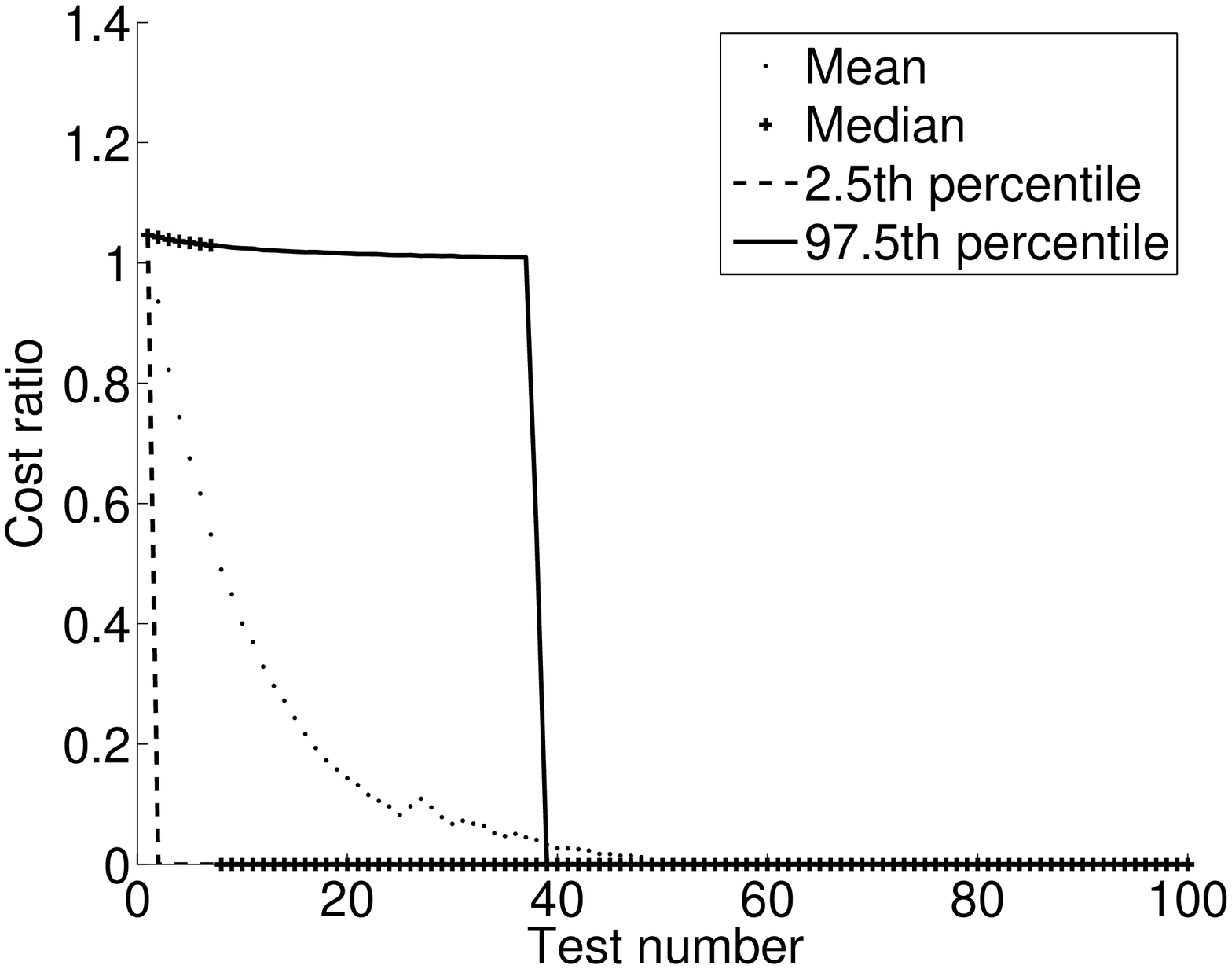}}
\label{fig_compare_c} }
\caption{\label{fig_compare}Simulation Results.
(a) shows the mean cost of the i'th test. (b) shows the mean, median, and 2.5th and 97.5th percentiles of the ratio between the cost of the i'th test in $\QPDASR$ and $\QPDAS$. (c) shows the same for $\QPDASROPT$ and $\QPDAS$.
}
\end{figure}

\section{Conclusions}
In this paper we discussed sequential procedures for controlling
false discovery, and their application to public databases.
Specifically, we concentrated on the only known sequential procedure
for controlling an $FDR$-like measure, Alpha Investing, and on the
simpler Alpha Spending that controls $FWER$. We have shown that both can
be described as special cases of a more general framework, which
offers a trade-off between the level of each test and the reward
obtained when the hypothesis is rejected.

We have further shown that from this general framework two additional,
novel procedures may be derived, which have practical benefits over
existing procedures.
The version we termed \ero~Alpha Investing
is optimal in the sense that, under certain practical assumptions, it tunes the balance between level
of the test and the reward such as to maximize the potential
for subsequent tests. We have shown it to be significantly more powerful
 than other Alpha Investing variants in practice. The second, Alpha Spending with Rewards, can be viewed as
a cross between Alpha Spending and Alpha Investing. Its practical
benefit is that systems that rely on Alpha Spending can be adapted to
work with Alpha Spending with Rewards with minimal changes, and enjoy
the additional power gained from controlling a less conservative
measure of type-I errors.

Finally, we have shown how Alpha Spending with Rewards can be
employed to control false discovery in public databases
through the framework of Quality Preserving Database (QPD). This
demonstrates our previous point, where the QPD framework, previously
relying on Alpha Spending, was seamlessly transformed to work with
Alpha Spending with Rewards, something that could not have been
done with the original Alpha Investing. In addition, this has achieved a
significant improvement to the QPD, since replacing $FWER$ with
$mFDR_\eta$ as the controlled measure resulted in costs reduced to
practically zero.

This last point is of crucial importance. It implies that a layer for
controlling false discovery such as the QPD can be implemented for
a freely accessible public database, keeping it almost freely
accessible still. The Generalized Alpha Investing procedure and
the $mFDR_\eta$ measure have some drawbacks, most notably the
required independence Assumption~\ref{ass_ind}, inherited from Alpha
Investing. Yet employing it as a method for false discovery
control for public databases comes at almost no cost, and it surely
offers a huge improvement over the current situation of no false
discovery control enforced for public databases to date. All the more
so since in the case of public databases, with different research groups
researching different aspects of the data, Assumption~\ref{ass_ind}
can often be expected to hold approximately.

\section{Appendix}

\subsection{Technical Results and Proofs}
\begin{lemma}
Given Assumption~\ref{ass_ind}, in the Generalized Alpha Investing Procedure
the stochastic process
$A(j)\equiv \alpha R(j)-V(j)+\alpha \eta - W(j)$ is a sub-martingale
with respect to $R_j$.\label{lem_martingale}\end{lemma}
\begin{proof}
We need to show
$E_\theta(A(j)|R_{j-1},R_{j-2},\ldots,R_1)\geq A(j-1)$.
Note that in analyzing this conditional expectation we will treat $\ccj$, $\ddj$,
and $\alpha_j$ as known constants, since given
$R_{j-1},R_{j-2},\ldots,R_1$ they are deterministically chosen
by rule~\ref{equ_gai_rule}.

We will use the abbreviation
$E_\theta^{j-1}(X)=E_\theta(X|R_{j-1},R_{j-2},\ldots,R_1)$.

Let us define $A_j=A(j)-A(j-1)$.
Since $R(j)-R(j-1)=R_j$, $V(j)-V(j-1)=V_j$ and
$W(j)-W(j-1)=-\ccj+R_j \ddj$,
we can rewrite $A_j$ more explicitly as
$A_j=(\alpha-\ddj)R_j - V_j +\ccj$.

Assuming a true null hypothesis, $\theta\in H_j$, we get that $R_j=V_j$. Hence
$A_j=-(1-\alpha+\ddj) R_j + \ccj$.
A series of inequalities achieves the desired result:
$E_\theta^{j-1}(A_j)  = \\
-(1-\alpha+\ddj) E_\theta^{j-1}(R_j)+\ccj \geq -(1-\alpha+\ddj) \alpha_j +\ccj
                \geq  -(1-\alpha+\frac{\ccj}{\alpha_j}+\alpha-1)\alpha_j+\ccj
               = 0$.\\
The first inequality follows from Assumption~\ref{ass_ind},
and because
$\ddj$ and $1-\alpha$ are nonnegative. The second inequality follows
because by equation~\ref{equ_gai_max_reward}
$\ddj \leq \frac{\ccj}{\alpha_j}+\alpha-1$,

Assuming a false null hypothesis, $V_j=0$, hence
$A_j=(\alpha-\ddj) R_j+\ccj$. Now the following series of inequalities will
achieve the desired result:
$E_\theta^{j-1}(A_j)  = (\alpha-\ddj) E_\theta^{j-1}(R_j)+\ccj
               \geq \\(\alpha-\frac{\ccj}{\ppj}-\alpha) E_\theta^{j-1}(R_j)+\ccj
               = -\frac{\ccj}{\ppj} E_\theta^{j-1}(R_j)+\ccj
               \geq  -\frac{\ccj}{\ppj} \ppj +\ccj
               =  0$.\\
The first inequality follows
from equation~\ref{equ_gai_max_reward} since
$\ddj \leq \frac{\ccj}{\ppj}+\alpha$, and the second inequality
is due to Assumption~\ref{ass_ind}.\qed
\end{proof}


\begin{theorem}
ERO Alpha Investing Optimality: Over a series of
tests where Assumption~\ref{regular_test} holds,
the \ero~Alpha Investing procedure (Definition~\ref{def_ero_alpha_investing}) is an \ero~procedure (Definition~\ref{def_expected_reward_optimal}).
\label{theorem_ero_full}
\end{theorem}
\begin{proof}
Clearly, once $\alpha_j$ is set, the optimal choice for $\ddj$ is the
maximal one allowed by constraint~\ref{equ_gai_max_reward}. All that remains
to show is that $\frac{\ccj}{\ppj}=\frac{\ccj}{\alpha_j}-1$ has
solutions and that they are all indeed optimal.
Lemma~\ref{lem_equ_solve} shows the former, and Lemma~\ref{lem_equ_opt} the latter.\qed
\end{proof}

\begin{assumption}
Let $\ppj$ be the best-power (according to Definition~\ref{def_best_power}),
and $\gamma_j$ be the actual probability of rejection for a given $\theta$, $\gamma_j=P_\theta(R_j=1)$.
Then we assume the following:  (1) $\ppj$ and $\gamma_j$ as functions of $\alpha_j$ are continuous, monotonically non-decreasing, (2) $\frac{\ppj}{\alpha_j}$ is monotonically non-increasing, (3) $\frac{\ppj}{\gamma_j}$ is monotonically non-increasing, (4) $\frac{\gamma_j}{\alpha_j}$ is monotonically non-increasing.
\label{regular_test}
\end{assumption}

\begin{lemma}
Given Assumption~\ref{regular_test},
$\frac{\ccj}{\ppj}=\frac{\ccj}{\alpha_j}-1$
is solvable in the range $0<\alpha_j<1$.
\label{lem_equ_solve}
\end{lemma}
\begin{proof}
Let us denote $f(\alpha_j)=\frac{\ccj}{\ppj}(1-\frac{\ppj}{\alpha_j})+1$.
It is enough to show $f(\alpha_j)=0$ for some $0<\alpha_j<1$.
By Assumption~\ref{regular_test}, we know $f(\alpha_j)$ is a continuous function over the range $(0,1]$, and clearly
$f(1)=1$. Thus, we can conclude the proof by showing $lim_{\alpha_j \to 0} f(\alpha_j)=-\infty$.
Since $\ppj$ is non-negative and non-decreasing, the limit $lim_{\alpha_j \to 0} \ppj$ exists. Let us denote it by $\delta$. If $\delta>0$, then $f(\alpha_j)<\frac{\ccj}{\ppj}(1-\frac{\delta}{\alpha_j})+1$, which implies
$lim_{\alpha_j \to 0} f(\alpha_j)=-\infty$.
If $\delta=0$, then by Assumption~\ref{regular_test}, $\frac{\ppj}{\alpha_j}$ is monotonically non-increasing with $\alpha_j$. Since $\frac{\ppj}{\alpha_j}\geq 1$, and since the alternative and null hypotheses must be distinct, we
must have for some $0<\alpha_j^\star<1$ and positive $\epsilon$ that $\forall \alpha_j<\alpha_j^\star: \frac{\ppj}{\alpha_j}\geq 1+\epsilon$. Therefore $f(\alpha_j)<-\epsilon \frac{\ccj}{\ppj}+1$. Since we assumed $\delta=0$, we have proved $lim_{\alpha_j \to 0} f(\alpha_j)=-\infty$.\qed
\end{proof}

\begin{lemma}
Given Assumption~\ref{regular_test}, then for
any choice of $\ccj$, and for any $\theta$, the expression
$E_\theta(R_j) min(\frac{\ccj}{\ppj}+\alpha,\frac{\ccj}{\alpha_j}+\alpha-1)$
is maximized by setting $\alpha_j$ to be the solution of
$\frac{\ccj}{\ppj}=\frac{\ccj}{\alpha_j}-1$.
\label{lem_equ_opt}
\end{lemma}
\begin{proof}
Lemma~\ref{lem_equ_solve} proves $\frac{\ccj}{\ppj}=\frac{\ccj}{\alpha_j}-1$ is solvable for some $0<\alpha_j^\star<1$.
We will now show that $\alpha_j^\star$ maximizes
$E_\theta(R_j) min(\frac{\ccj}{\ppj}+\alpha,\frac{\ccj}{\alpha_j}+\alpha-1)$.
Using the notation of Assumption~\ref{regular_test}, we can write $\gamma_j=E_\theta(R_j)$. Note that $\alpha_j^\star$ brings the two arguments of the $min$ operator to equality, therefore it suffices to show that
$\gamma_j(\frac{\ccj}{\ppj}+\alpha)$ is monotonically non-decreasing with $\alpha_j$ and that
$\gamma_j(\frac{\ccj}{\alpha_j}+\alpha-1)$ is monotonically non-increasing.
$\gamma_j(\frac{\ccj}{\ppj}+\alpha)$ can be rewritten as
$\ccj\frac{\gamma_j}{\ppj}+\gamma_j\alpha$. This is monotonically non-decreasing, since $\ccj,\alpha$ are constants,
and by Assumption~\ref{regular_test},
$\gamma_j$ is non-decreasing,
and $\frac{\gamma_j}{\ppj}$ is non-decreasing.
$\gamma_j(\frac{\ccj}{\alpha_j}+\alpha-1)$ can be rewritten as
$\ccj\frac{\gamma_j}{\alpha_j}-(1-\alpha)\gamma_j$. This is monotonically non-increasing, since
$\ccj,\alpha$ are constants, and by Assumption~\ref{regular_test},
$\gamma_j$ is non-decreasing, $\frac{\gamma_j}{\alpha_j}$ is non-increasing.\qed
\end{proof}

\begin{lemma}
Assumption~\ref{regular_test} holds for Neyman-Pearson tests with continuous distribution functions.
\label{lem_pearson_satisfies}
\end{lemma}
\begin{proof}
The continuity of $\ppj$ and $\gamma_j$ is implied by the underlying continuous distribution functions.
Clearly $\ppj$ is a non-decreasing function of $\alpha_j$.
$\frac{\ppj}{\alpha_j}$ is monotonically non-increasing by Lemma~\ref{lem_pearson}.
In Neyman-Pearson tests, $\gamma_j$ is either identical to $\ppj$ or to $\alpha_j$, hence
$\frac{\ppj}{\gamma_j}$ is either the identity function or $\frac{\ppj}{\alpha_j}$. In either case it is monotonically non-increasing.
A similar argument holds for $\frac{\gamma_j}{\alpha_j}$ and for showing $\gamma_j$ is non-decreasing.\qed
\end{proof}

\begin{lemma}
For a uniformly most powerful test with a continuous distribution function, a single parameter $\theta$, and a simple null hypothesis $H_0: \theta=\theta_0$, Assumption~\ref{regular_test} holds.
\label{lem_ump_satisfies}
\end{lemma}
\begin{proof}
The continuity of $\ppj$ and $\gamma_j$ is implied by the underlying continuous distribution functions.
If the alternative hypothesis is unbounded, then $\ppj\equiv 1$, for which all the required properties involving $\ppj$ follow trivially. If the alternative hypothesis is bounded, then the best power $\ppj$ is calculated at some extreme point $\theta_1$. According to the Neyman-Pearson lemma, the power is identical to the power of a test with a simple alternative $H_1:\theta=\theta_1$, hence Lemma~\ref{lem_pearson_satisfies} proves all the required properties involving $\ppj$.
Regarding $\gamma_j$, it is computed at the actual $\theta$. If $H_0$ happens to be true, then $\gamma_j\equiv \alpha_j$, and we trivially get all the properties of $\gamma_j$. Otherwise, a similar argument using the Neyman-Pearson lemma applies.\qed
\end{proof}

\begin{lemma}
In Neyman-Pearson tests with continuous distribution functions with level $\alpha$ and power $\rho$,
$\frac{\rho}{\alpha}$ is a
monotonically non-increasing function of $\alpha$.
\label{lem_pearson}
\end{lemma}
\begin{proof}
By monotonicity of likelihood ratio implied by the Neyman-Pearson
lemma, this is guaranteed.\qed
\end{proof}


\begin{table}
  \caption{\label{table_constant2}Simulation results for the 'constant'
    allocation scheme.}
  \centering
\fbox{%
\begin{tabular}{p{6cm} | p{1cm} | p{1cm} | p{1cm} | p{1cm}}
\em Procedure       & \em Tests        & \em True rejects & \em False rejects & \em mFDR \\
\hline
Alpha Spending &       10.000 &        0.000 &        0.046 &        0.046 \\
Alpha Investing &       10.563 &        0.000 &        0.047 &        0.047 \\
Alpha Spending with Rewards (k=1) &       10.513 &        0.000 &        0.046 &        0.047 \\
Alpha Spending with Rewards (k=1.1) &       11.543 &        0.000 &        0.048 &        0.048 \\
\ero~Alpha Investing &       17.152 &        0.000 &        0.018 &        0.018 \\
\end{tabular}}
\end{table}

\begin{table}
  \caption{\label{table_relative2}Simulation results for the 'relative' allocation scheme.}
  \centering
\fbox{%
\begin{tabular}{p{6cm} | p{1cm} | p{1cm} | p{1cm} | p{1cm} | p{1cm}}
\em Procedure       & \em Tests        & \em True rejects & \em False rejects & \em mFDR \\
\hline
Alpha Spending &       66.000 &        0.001 &        0.048 &        0.048 \\
Alpha Investing &       66.678 &        0.001 &        0.047 &        0.047 \\
Alpha Spending with Rewards (k=1) &       66.665 &        0.001 &        0.047 &        0.047 \\
Alpha Spending with Rewards (k=1.1) &       67.019 &        0.001 &        0.049 &        0.049 \\
\ero~Alpha Investing &       66.925 &        0.001 &        0.056 &        0.056 \\
\end{tabular}}
\end{table}

\begin{table}
\caption{\label{table_rel2002}Simulation results for the 'relative-200' allocation scheme.}
  \centering
\fbox{%
\begin{tabular}{p{6cm} | p{1cm} | p{1cm} | p{1cm} | p{1cm} | p{1cm}}
\em Procedure & \em Tests & \em True rejects & \em False rejects & \em mFDR \\
\hline
Alpha Spending &      200 &        0.000 &        0.051 &        0.051 \\
Alpha Investing &      200 &        0.000 &        0.052 &        0.052 \\
Alpha Spending with Rewards (k=1) &      200 &        0.000 &        0.052 &        0.052 \\
Alpha Spending with Rewards (k=1.1) &      200 &        0.001 &        0.054 &        0.053 \\
\ero~Alpha Investing &      200 &        0.001 &        0.055 &        0.055 \\
\end{tabular}}
\end{table}

\begin{lemma}
In $\QPDASR(\alpha,\eta,q)$ it holds that $W(j)\geq \alpha\eta q^{n_j-n_0}$.
\label{lem_qpdr_superior}
\end{lemma}
\begin{proof}
The proof is by induction.
According to equations~\ref{equ_asr_update},
$W(0)=\alpha\eta$ (Recall that we use $k=1$ for simplicity of
notation).
The following series of inequalities conclude the proof using
the induction hypothesis,
and equations~\ref{equ_asr_update},~\ref{equ_persistent_level}.
$
W(j)  = W(j-1)-\alpha_j + \alpha R_j \geq  \\
W(j-1)-\alpha_j =  W(j-1)-W(j-1)(1-q^{c_j}) =
W(j-1) q^{c_j}
 \geq  \alpha \eta q^{n_{j-1}-n_0} q^{c_j}
 =  \alpha\eta q^{n_j-n_0}.
$

\end{proof}

\subsection{Generalized Alpha Investing: Additional Simulation Results}
Tables~\ref{table_constant2},~\ref{table_relative2}, and~\ref{table_rel2002} depict the
simulation results with a frequency of $1\%$ false null hypotheses, and effect size of $0.1$.
Regarding number of tests, \ero~Alpha Investing significantly out-performed all competitors in both 'constant' and 'relative' allocation schemes (paired t-test p-values at most 0.0056), except when compared with the Alpha Spending with Rewards with k=1.1, in the 'relative' allocation scheme, where the difference was not significant (p-value 0.28).
Regarding number of true-rejections, all procedures are performing poorly in these extreme conditions. Hence no significant differences could be found here.
Due to the extremely low amount of rejections, the estimation of $mFDR$ is inaccurate, hence in some cases it seems above $0.05$.

\section{Acknowledgments}
We would like to thank
Yoav Benjamini, Marina Bogomolov, Ruth Heller, Daniel Yekutieli,
Hani Neuvirth, and the reviewing team
for their thoughtful and useful comments.
This work was supported by an Open Collaborative Research grant from
IBM, and by Israeli Science Foundation grant ISF-1227/09.

\bibliographystyle{plain}
\bibliography{paper}

\begin{thebibliography}{10}

\bibitem{qpd2010}
Ehud Aharoni, Hani Neuvirth, and Saharon Rosset.
\newblock The quality preserving database: A computational framework for
  encouraging collaboration, enhancing power and controlling false discovery.
\newblock {\em Computational Biology and Bioinformatics, IEEE/ACM Transactions
  on}, 8(5):1431--1437, 2011.

\bibitem{YimingBao01152008}
Yiming Bao, Pavel Bolotov, Dmitry Dernovoy, Boris Kiryutin, Leonid Zaslavsky,
  Tatiana Tatusova, Jim Ostell, and David Lipman.
\newblock The influenza virus resource at the national center for biotechnology
  information.
\newblock {\em Journal of virology}, 82(2):596--601, 2008.

\bibitem{fdr}
Yoav Benjamini and Yosef Hochberg.
\newblock Controlling the false discovery rate: a practical and powerful
  approach to multiple testing.
\newblock {\em Journal of the Royal Statistical Society: Series B (Statistical
  Methodology)}, pages 289--300, 1995.

\bibitem{false_discovery_dependency}
Yoav Benjamini and Daniel Yekutieli.
\newblock The control of the false discovery rate in multiple testing under
  dependency.
\newblock {\em The Annals of Statistics}, 29(4):1165--1188, 2001.

\bibitem{alphainvesting}
Dean~P. Foster and Robert~A. Stine.
\newblock Alpha-investing: a procedure for sequential control of expected false
  discoveries.
\newblock {\em Journal of the Royal Statistical Society: Series B (Statistical
  Methodology)}, 70(2):429--444, January 2008.

\bibitem{hardy1}
John Hardy and Andrew Singleton.
\newblock Genomewide association studies and human disease.
\newblock {\em New England Journal of Medicine}, 360(17):1759--1768, 2009.

\bibitem{gwas_review}
Joel~N. Hirschhorn and Mark~J. Daly.
\newblock Genome-wide association studies for common diseases and complex
  traits.
\newblock {\em Nature reviews. Genetics}, 6(2):95--108, February 2005.

\bibitem{ion}
John~PA Ioannidis.
\newblock Why most published research findings are false.
\newblock {\em PLoS medicine}, 2(8):e124, 2005.

\bibitem{kruglyak1}
E.~Lander and L.~Kruglyak.
\newblock {Genetic dissection of complex traits: guidelines for interpreting
  and reporting linkage results.}
\newblock {\em Nature genetics}, 11(3):241--247, November 1995.

\bibitem{naj1}
Adam~C. Naj, Gyungah Jun, Gary~W. Beecham, Li-San~S. Wang, Badri Narayan~N.
  Vardarajan, et~al.
\newblock {Common variants at MS4A4/MS4A6E, CD2AP, CD33 and EPHA1 are
  associated with late-onset Alzheimer's disease.}
\newblock {\em Nature genetics}, 43(5):436--441, May 2011.

\bibitem{promateus}
Hani Neuvirth, Uri Heinemann, David Birnbaum, Naftali Tishby, and Gideon
  Schreiber.
\newblock Promateus--an open research approach to protein-binding sites
  analysis.
\newblock {\em Nucleic acids research}, 35(suppl 2):W543--W548, 2007.

\bibitem{HIVdb}
Soo-Yon Rhee, Matthew~J Gonzales, Rami Kantor, Bradley~J Betts, Jaideep Ravela,
  and Robert~W Shafer.
\newblock Human immunodeficiency virus reverse transcriptase and protease
  sequence database.
\newblock {\em Nucleic acids research}, 31(1):298--303, 2003.

\bibitem{bias}
Robert~John Simes.
\newblock Publication bias: the case for an international registry of clinical
  trials.
\newblock {\em Journal of Clinical Oncology}, 4:1529--1541, 1986.

\bibitem{storey_fdr}
John~D Storey.
\newblock A direct approach to false discovery rates.
\newblock {\em Journal of the Royal Statistical Society: Series B (Statistical
  Methodology)}, 64(3):479--498, 2002.

\bibitem{gene_combine}
Hua Tang, Jie Peng, Pei Wang, Marc Coram, and Li~Hsu.
\newblock Combining multiple family-based association studies.
\newblock {\em BMC Proceedings}, 1(Suppl 1):S162, 2007.

\bibitem{gene_false}
Edwin J. C.~G. van~den Oord and Patrick~F. Sullivan.
\newblock False discoveries and models for gene discovery.
\newblock {\em Trends in Genetics}, 19(10):537--542, 2003.

\bibitem{WTCCC}
{Wellcome Trust Case Control Consortium}.
\newblock Genome-wide association study of 14,000 cases of seven common
  diseases and 3,000 shared controls.
\newblock {\em Nature}, 447(7145):661--678, June 2007.

\bibitem{wellcome2}
{Wellcome Trust Case Control Consortium}.
\newblock {Genome-wide association study of CNVs in 16,000 cases of eight
  common diseases and 3,000 shared controls}.
\newblock {\em Nature}, 464(7289):713--720, April 2010.

\end{thebibliography}
\end{document}